\documentclass[journal,onecolumn]{IEEEtran}

\usepackage{times}
\usepackage{amsmath,amssymb,amsfonts}
\usepackage{graphicx}
\usepackage{textcomp}
\usepackage{algorithm}
\usepackage{algorithmic}
\usepackage{bbm}

\usepackage{comment}
\usepackage{colortbl,color}
\usepackage{latexsym}
\usepackage{dsfont}
\usepackage{amsthm}
\usepackage{float}
\usepackage{hhline}
\usepackage{graphicx}
\usepackage{verbatim}
\usepackage{eurosym}
\usepackage{dsfont}
\usepackage{booktabs}
\newtheorem{theorem}{Theorem}
\newtheorem{lemma}[theorem]{Lemma}
\newtheorem{corollary}[theorem]{Corollary}

\newtheorem{definition}{Definition}
\newtheorem{assumption}{Assumption}

\begin{document}

\title{Linear Quadratic Control with Non-Markovian and Non-Semimartingale Noise Models}

\author{Mostafa~M.~Shibl,~\IEEEmembership{Graduate~Student~Member,~IEEE,} Sharan Srinivasan, Harsha Honnappa, and Vijay Gupta,~\IEEEmembership{Fellow,~IEEE}
\thanks{M. M. Shibl, S. Srinivasan, and V. Gupta are with the Elmore Family School of Electrical and Computer Engineering, Purdue University, West Lafayette, IN 47906, USA         {\tt\small \{mabdelna,srini256,gupta869\}@purdue.edu}.}
\thanks{H. Honnappa is with the School of
Industrial Engineering, Purdue University, West Lafayette, IN 47906, USA {\tt\small honnappa@purdue.edu}.}
}

\maketitle

\begin{abstract}
The standard linear quadratic Gaussian (LQG) framework assumes a Brownian noise process and relies on classical stochastic calculus tools, such as those based on Itô calculus. In this paper, we solve a generalized linear quadratic optimal control problem where the process and measurement noises can be non-Markovian and non-semimartingale stochastic processes with sample paths that have low Hölder regularity. Since these noise models do not, in general, permit the use of the standard Itô calculus, we employ rough path theory to formulate and solve the problem. By leveraging signature representations and controlled rough paths, we derive 
the optimal state estimation and control strategies. 
\end{abstract}

\begin{IEEEkeywords}
Generalized linear quadratic control, rough differential equations, rough paths theory, non-Markovian and non-semimartingale noise
\end{IEEEkeywords}

\IEEEpeerreviewmaketitle

\section{Introduction}
\label{sec:intro}
\IEEEPARstart{L}{i}near quadratic Gaussian (LQG) optimal control is a fundamental problem in stochastic control theory and widely used in various engineering applications. The standard LQG framework assumes that the system dynamics and observations are corrupted by Brownian noise processes  allowing for the derivation of a separation principle. The optimal controller can then be obtained through the application of Kalman filtering for optimal state estimation and using the Hamilton-Jacobi-Bellman (HJB) equation to derive the optimal feedback control. The analytical machinery for each of these relies on the standard Itô stochastic calculus. For theoretical background on the standard LQG problem, we refer the reader to~\cite{book5}. 

However, empirical measurements of various physical phenomena often indicate that Brownian processes are not always the most effective modeling choice. Many real-world processes exhibit dependencies, correlations, and sample paths that have low Hölder regularity, which may not be adequately modeled by standard Brownian processes, and instead require non-Markovian and/or non-semimartingale models. For instance, due to its ability to capture memory effects and self-similarity properties, empirically, fractional Brownian motion (fBm) has been shown to better model a broad range of phenomena such as in financial systems~\cite{ted}, turbulent systems~\cite{turbulence1,turbulence2}, rainfall systems~\cite{rainfall}, and biological systems~\cite{biology}. Other examples include financial data modeled as stable L\'evy processes~\cite{financial} and noise processes in quantum systems  modeled as shot noise~\cite{quantum}.

Moving from Brownian processes to these more general processes in an estimation and control context poses theoretical challenges. Unlike Brownian motion, more general processes that may be non-Gaussian, non-Markovian and/or non-semimartingale require more analysis, such as careful characterization of more attributes like tail distributions, skewness, higher order moments, and dependence and correlation parameters. Ignoring these effects and forcing LQG controllers that rely on the noise processes being modeled as Brownian processes can significantly impact system behavior and control performance, even leading to system failures~\cite{qz}.

The extension of LQ control to systems affected by noise sequences that are more general than Brownian noise processes is, thus, an important area of research. One model that has seen significant work is when the noise is still Gaussian; however, instead of white Brownian noise, we consider a colored noise sequence. In this context,~\cite{bet} investigates stochastic LQR optimal control with both white and colored Gaussian noises, employing dynamic programming techniques to derive optimal solutions. A traditional approach here is to transform the colored noise into white noise through appropriate filtering, thereby allowing the application of standard LQG methods. Some recent studies considering such models include~\cite{jml}, which proposes the use of shaping filters for systems with colored measurement noises,~\cite{jh2}, which explores how pink noise, a type of colored noise, affects optimal policies in reinforcement learning (highlighting the deviations from assumption of a white noise), and~\cite{fas}, which demonstrates the application of optimal vibration control in the presence of colored noise, showcasing practical implications of such models.

There is a line of work that studies control design with non-Markovian, but semimartingale noise, which allows the use of standard Itô stochastic calculus. For example,~\cite{as} addressed detection and estimation problems in heavy-tailed symmetric alpha-stable distributed noise,~\cite{ted3} addressed stochastic control in the presence of Rosenblatt process noise,~\cite{st} provided the necessary conditions for optimal control of stochastic systems with Markovian random jump noise models, and~\cite{yh} extended this work to partial information LQ control for jump diffusions, which can have non-Markovian coefficients for drift, diffusion, and jump terms.

A particular noise process that has attracted much attention is fractional Brownian motion (fBm), which is a Gaussian process that is, in general, non-Markovian, as seen in the works of~\cite{mlk},~\cite{tl2} and~\cite{mlk2}. When the Hurst index exceeds 0.5, classical Itô stochastic calculus needs to be extended to the use of fractional stochastic calculus in these papers. Fractional calculus deals with derivatives and integrals of arbitrary, non-integer orders~\cite{bg}. Several works have studied optimal LQR control design for a system driven by fBm, including~\cite{ted},~\cite{ted2},~\cite{yh2} and~\cite{ted4}. 

However, when dealing with non-semimartingale noise processes with low Hölder regularity, the complexity of the mathematical analysis increases. Non-semimartingale processes cannot be decomposed into the sum of a local martingale and a finite variation process. Thus, when such noise models are used, Itô calculus is inapplicable. This includes fBm with Hurst index less than 0.5. We call such problems as \textit{generalized LQ or gLQ problems}. Unlike standard Brownian motion, when considering non-semimartingale noise models with low Hölder regularity, it is not clear whether classical Kalman filtering and LQG control results are applicable or even how to define terms in classical LQG. 

When considering low Hölder regularity non-semimartingale noise sequences, the lack of an Itô integral necessitates alternative mathematical frameworks to be able to analyze systems driven by such noise sequences. To analyze such systems, we employ rough path theory, which encompasses and generalizes classical stochastic calculus. Rough path theory provides a toolbox for the analysis of stochastic differential equations driven by non-semimartingale processes and rough signals that have low Hölder regularity~\cite{book,book2,book3}.~\cite{jd} provided a foundational contribution by exploring stochastic control using rough path theory when considering standard Brownian motion noise (i.e. Hölder regularity $H = 1/2$). That paper formulated controlled rough stochastic differential equations and demonstrated that the value function satisfies a Hamilton-Jacobi-Bellman (HJB) type equation, and provides a continuous map for the solution. This continuity property may not, in general, hold with standard Itô calculus. This has been followed up in~\cite{lew2025rough} where a rough stochastic Pontryagin principle was derived, as well as optimal relaxed control theory with rough differential equations in~\cite{ashkarian2026pontryagin,chakraborty2024pathwise}.  In all of these settings, rough path analysis can provide a continuous pathwise deterministic solution to stochastic differential equations driven by rough signals. In addition, in the case of fractional noises such as fBm, fractional Stratonovich integrals can be approximated using Wong-Zakai and mollifier approximations, as discussed by~\cite{ct}, for rough differential equations. More recently there has also been significant interest in exploiting signature-based methods for optimal control~\cite{bayer2023optimal,bank2025stochastic}. Nonetheless, very little attention has been paid to control design questions.

In this paper, we introduce the gLQ setting, where we have an LQ optimal control problem formulation with the process noise and measurement noise modeled as non-Markovian and non-semimartingale sequences, with Hölder regularity $H \in (1/3,1]$. It should be noted that this includes the special class of stochastic processes that are highly irregular with Hölder regularity $H \in (1/3,1/2)$. Due to the consideration of non-semimartingale processes and low Hölder regularity, we cannot utilize classical Itô calculus. Thus, we employ rough path theory for the analysis of controlled rough differential equations to provide a mathematically rigorous framework for solving the gLQ problem. By leveraging the rough path framework, we formulate a well-posed gLQ control problem, and derive an optimal control policy and an optimal observer that account for the rough process and measurement noises. First, we explicitly characterize the optimal control policy, considering perfect state availability, and discuss its differences from the classical LQ solution. Second, we characterize the optimal observer, including the cases of independent and correlated process and measurement noises. Finally, we implement the proposed approach in a simulation environment to illustrate its effectiveness in practical scenarios. 

The paper is organized as follows. Section~\ref{sec:rough_path_theory} introduces some preliminaries on rough path theory. Section~\ref{sec:model} presents the model and formulates the problem considered. Next, Section~\ref{sec:math} provides the main results. Section~\ref{sec:example} illustrates the results with two numerical examples. Finally, Section~\ref{sec:conc} concludes the paper and proposes some future directions. Appendix~\ref{app1} presents some supporting results, while Appendix~\ref{app2} presents some examples of noise processes that can be considered using our tools, but not with the traditional Itô calculus.

\section{Rough Path Theory Preliminaries}
\label{sec:rough_path_theory}
As previously mentioned, since we consider process noise and measurement noise modeled as non-Markovian and non-semimartingale sequences, with low Hölder regularity index $H \in (1/3,1]$, we will use rough path theory to guide our control design and synthesis.

Definition~\ref{def:rough_path} formally defines the rough path lift.

\begin{definition}[Rough Paths~\cite{book}]
\label{def:rough_path}
Let \( V \) be a Banach space and \( [0,T] \subset \mathbb{R} \). 
A rough path lift with Hölder regularity $H \in (\frac{1}{3},\frac{1}{2}]$ over \( V \) on the interval \( [0,T] \) 
is a pair \( (X, \mathbb{X}) \) consisting of:
\begin{itemize}
    \item a path \( X : [0,T] \to V \),
    \item a second-order process \( \mathbb{X} : [0,T]^2 \to V \otimes V \),
\end{itemize}
such that for all \( 0 \le s \le u \le t \le T \), the following Chen's relations hold:
\begin{align*}
X_{s,t} &= X_{s,u} + X_{u,t}, \\
\mathbb{X}_{s,t} &= \mathbb{X}_{s,u} + \mathbb{X}_{u,t} + X_{s,u} \otimes X_{u,t}.
\end{align*}
Based on this definition, the rough path lift is a map from the simplex $\{(s,t): 0 \le s \le t \le T \}$ to the tensor algebra $T^{(2)}(V)$. Note that using Chen's relations on the first component only, we get the additive property of integrals.
\end{definition}

Definition~\ref{def:holder} defines the notion of Hölder regularity for a rough path.

\begin{definition}[Hölder Regularity of a Rough Path~\cite{book}]
\label{def:holder}
Let $\alpha \in (0,1]$ and $N = \lfloor 1/\alpha \rfloor$. 
A \emph{(geometric) $\alpha$-Hölder rough path} over $\mathbb{R}^d$ 
is a collection of maps
\begin{align*}
\mathbf{X}_{s,t} &= \big( X_{s,t}^{(1)}, X_{s,t}^{(2)}, \ldots, X_{s,t}^{(N)} \big),
\quad 0 \le s \le t \le T, \\
\textbf{X}_{s,t} &= \Bigg(\int_{s<s_1<t}dX(s_1), \int_{s<s_1<s_2<t}dX(s_1)\otimes dX(s_2),..., \int_{s<s_1<...<s_N<t}dX(s_1)\otimes ...\otimes dX(s_N)\Bigg).
\end{align*}
such that for each $k = 1, \ldots, N$ there exists a constant $C > 0$ satisfying
\[
  \| X_{s,t}^{(k)} \| \le C |t-s|^{k\alpha}, 
  \quad \forall\, s,t \in [0,T].
\]
Here, $X_{s,t}^{(1)} = x_t - x_s$ is the increment of the path, 
and higher-order terms $X_{s,t}^{(k)}$ represent iterated integrals satisfying the Chen's relations.
\end{definition}

Definition~\ref{def:geometric_lift} defines the notion of a geometric rough path lift.

\begin{definition}[Geometric $p$-rough path]
\label{def:geometric_lift}
Let $p\in[2,3)$ and let $V$ be a finite dimensional Banach space. 
A \emph{geometric $p$-rough path} over $V$ is a pair 
\[
\mathbf{X} = (X,\mathbb{X}) : 0\le s \le t \le T
\]
such that the following properties hold:
\begin{enumerate}
    \item \textbf{Finite $p$-variation.}  
    The first level $X$ has finite $p$-variation and the second level 
    $\mathbb{X}$ has finite $p/2$-variation, i.e.
    \[
    \|X\|_{p\text{-var};[0,T]} < \infty,
    \qquad
    \|\mathbb{X}\|_{p/2\text{-var};[0,T]} < \infty .
    \]
    \item \textbf{Chen's relation.}  
    For all $0 \le s \le u \le t \le T$,
    \[
    \mathbb{X}_{s,t}
        = \mathbb{X}_{s,u} + \mathbb{X}_{u,t}
          + X_{s,u} \otimes X_{u,t}.
    \]
    \item \textbf{Geometricity.}  
    There exists a sequence of smooth (bounded variation) paths 
    $\{X^{(n)}\}$ such that their lifts converge to $\mathbf{X}$ in the $p$-variation rough path metric.
\end{enumerate}
\end{definition}

Rough path lifts are done with respect to a given sample path, that is the integrations are done with respect to a sample path rather than a stochastic process. In the case of geometric rough paths, they are the closure of the lift of smooth paths, and can thus be approximated by smooth approximations.

A function $x\colon [0,T]\to\mathbb{R}^n$ is said to be a controlled rough path with Hölder regularity $H$ if there exists a path $x'(t)$ such that for all $0\leq s\leq t\leq T$,
\[
x(t)=x(s)+x'(s)(v(t)-v(s))+R_{s,t},
\]
with $\|R_{s,t}\|=O(\vert t-s \vert^{2H})$. By the Sewing Lemma (Lemma~\ref{lemma:sewing}), the rough integral
\[
\int_0^T x(s)\,d\mathbf{v}(s)
\]
is well defined under the condition that the path $x'(t)$ is unique. The uniqueness holds if one of the following conditions are satisfied~\cite{gubinelli}:
\begin{itemize}
    \item $x(t)$ has increments that do not scale better than $|t-s|^{2H}$
    \item $x(t)$ satisfies the controlled rough path decomposition, which follows from the uniqueness of the Young integral and controlled rough paths spaces,
    \item $x(t)$ is Hölder continuous with exponent \(\alpha\) and \(\mathbf{v}(t)\) is a rough path with Hölder regularity \(\beta\) satisfying \((2+\alpha)\beta > 1\).
\end{itemize}

Moreover, by Lyons' universal limit theorem (Lemma~\ref{lemma:lyonsuniversallimit}), the rough differential equation
\begin{equation}
\label{eq:system_lifted}
dx(t)= \textbf{A}\,x(t)\,dt + \textbf{B}\,u(t)\,dt + d\mathbf{v}(t)
\end{equation}
admits a unique solution~\cite{lyons}. Due to the use of Lyons' rough path theory framework, we consider classes of stochastic processes that have rough path lifts, and the lifts are well defined using Lyons' rough path framework. Those include classical stochastic processes such as semimartingales and Markov processes, as well as other stochastic processes such as fBm~\cite{lyons}.

\section{Model}
\label{sec:model}
We consider an infinite-horizon gLQ problem for a continuous-time, time-invariant system, with potentially non-Markovian and non-semimartingale process noise that has low Hölder regularity. Thus, the classical stochastic differential equation 
\[
dx(t) = \textbf{A} x(t) dt + \textbf{B} u(t) dt + dv(t), \qquad x_0 \text{ given},
\]
does not admit a solution. Therefore, we use the rough path lift to describe the system dynamics using a rough differential equation as
\begin{equation}
\label{eq:system}
\begin{aligned}
dx(t) &= \textbf{A} x(t) dt + \textbf{B} u(t) dt + d\textbf{v}(t), \qquad x_0 \text{ given},
\end{aligned}
\end{equation}
where $x(t) \in \mathbb{R}^n$ is the state vector at time $t$, $u(t) \in \mathbb{R}^m$ is the control input vector at time $t$,  $\textbf{v}(t) \in \mathbb{R}^n$ is the (geometric) rough path process noise vector at time $t$, $x_0$ is the initial state, and $\textbf{A} \in \mathbb{R}^{n \times n}$ and $\textbf{B} \in \mathbb{R}^{n \times m}$ are the system transition matrices. The measured output of the system is given by
\begin{equation}
\label{eq:output}
y(t) = \textbf{C} x(t) + w(t),
\end{equation}
where 
$\textbf{C} \in \mathbb{R}^{p \times n}$ is the output matrix, $y(t) \in \mathbb{R}^p$ is the output measurement vector at time $t$, and $w(t) \in \mathbb{R}^p$ is a measurement noise vector at time $t$. The measurement noise is also assumed to be potentially non-Markovian and non-semimartingale with Hölder regularity $H \in (\frac{1}{3},1]$. We utilize the standard quadratic performance criterion given by
\begin{equation}
\label{eq:cost}
J(u) = \lim_{T\rightarrow\infty}\frac{1}{T}\mathbb{E}\left[ \int_{0}^{T} \left( x(t)^\top \textbf{Q}\, x(t) + u(t)^\top \textbf{R}\, u(t) \right)\right] dt,
\end{equation}
where the expectation is taken over the process and measurement noises, and $\textbf{Q} \in \mathbb{R}^{n \times n}$ and $\textbf{R} \in \mathbb{R}^{m \times m}$ are the state weighting and input weighting matrices, respectively. We make the following standard assumptions:
\begin{assumption}
\label{ass:mainassumption}
We assume that (i) $(\textbf{A},\textbf{B})$ is stabilizable, (ii) $(\textbf{A},\textbf{C})$ is detectable, (iii) $\textbf{Q}$ and $\textbf{R}$ are symmetric and positive definite matrices, (iv) $(\textbf{A},\textbf{Q}^{1/2})$ is detectable,  and (v) the noise processes $v(t)$ and $w(t)$ have bounded covariances. 
\end{assumption}
The assumption about bounded covariances of the noise processes ensures that, with a suitable control design, the system is mean-square stable and the cost function is well-defined.

We are interested in designing a controller for the system above. In general, the control input $u(t)$ at any time $t$ is a function of the measurements $\{y(s):0\leq s\leq t\}.$ If the state $x(t)$ is not available for measurement, we will utilize an observer, described by a rough differential equation, of the form
\begin{equation}
\label{eq:observer}
\begin{aligned}
d\hat{x}(t) &= \textbf{A} \hat{x}(t)dt + \textbf{B} u(t)dt + \textbf{L}\hat{y}(t) dt- \textbf{LC}x(t) dt - \textbf{L}d\textbf{w}(t), \quad \hat{x}_0 \text{ given} \\
\hat{y}(t) &= \textbf{C} \hat{x}(t).
\end{aligned}
\end{equation}
where $\hat{x}(t)$ denotes the estimate of the state, $\hat{x}_0$ is the initial condition of the filter, $\textbf{w}(t)$ is the (geometric) rough path measurement noise vector at time $t$, and $\textbf{L}$ is the observer gain matrix to be designed.

Our objective is to design an optimal control law to calculate the control input $u(t)$ at any time $t\in[0,T]$, with the corresponding observer if needed, such that the cost functional \eqref{eq:cost} is minimized. As seen in (\ref{eq:system}) and (\ref{eq:observer}), we lift the noise processes $v(t)$ and $w(t)$ to rough paths, denoted by $\mathbf{v}$ and $\mathbf{w}$, which include the paths and their corresponding iterated integrals (or L\'evy areas). This enhancement then allows us to define integrals such as
\[
\int_{0}^{T} x(s)\,d\mathbf{v}(s)
\]
via the Sewing Lemma (stated for completeness as Lemma~\ref{lemma:sewing} in Appendix~\ref{app1}) to develop a rigorous derivation of the optimal controller.

\section{Main Results}
\label{sec:math}
We derive our result in two steps. First, we consider the LQR case, where $\textbf{C} = \textbf{I}$ and $w(t) = 0$ identically. This means the controller has full access to the state vector $x(t)$. In Section~\ref{sec:LQG}, we then consider the case when the controller has access only to the observations $y(t)$ and utilizes an observer.

\subsection{Optimal Controller with State Observation}
Lemma~\ref{lemma:rough_hjb} provides the definition of the solution of the rough HJB equation using the classical HJB equation~\cite{allan}.

\begin{lemma}[Definition 3.13 in~\cite{allan}]
\label{lemma:rough_hjb}
Let $\boldsymbol{\zeta}=(\zeta,\zeta^{(2)})$ be a geometric $p$-rough path on $[0,T]$ with $p\in(2,3)$. Let $\{\boldsymbol{\eta}^n\}_{n\ge1}$ be any sequence of lifted smooth paths $\boldsymbol{\eta}^n$ whose lifts $(\eta^n,(\eta^n)^{(2)})$ converge to $\boldsymbol{\zeta}$ in the $1/p$-Hölder rough path metric:
\[
\|\,\boldsymbol{\eta}^n - \boldsymbol{\zeta}\,\|_{1/p\text{-H{\"o}l}} \longrightarrow 0.
\]

For each smooth $\eta^n$, let $\Psi^{\eta^n}$ be the solution of the classical HJB equation. We say that a function $\Psi$ is a solution of the rough HJB equation driven by $\boldsymbol{\zeta}$ if
\[
\Psi^{\eta^n} \longrightarrow \Psi,
\]
for every sequence $\{\boldsymbol{\eta}^n\}$ converging to $\boldsymbol{\zeta}$ in the 
$1/p$-Hölder rough path topology.
\end{lemma}

Theorem~\ref{theorem:glqr} introduces the characterization of the optimal controller in the generalized LQR setting, with complete state information, meaning the controller has access to the full state vector.

\begin{theorem}[Existence \& Optimality of gLQ Controller]
\label{theorem:glqr}
Consider the problem formulation in Section~\ref{sec:model} such that $\textbf{C} = \textbf{I}$ and $w(t) = 0$. Under Assumption~\ref{ass:mainassumption}, there exists an optimal control $u^*(t)$ of the form
\[
u^*(t) = -\textbf{R}^{-1}\textbf{B}^\top \textbf{P} \left(x(t) + V(t)\right),
\]

where
\begin{enumerate}
    \item $\textbf{P}$ is the unique positive definite solution of the algebraic Riccati equation
    \[
    \textbf{A}^\top \textbf{P} + \textbf{P}\,\textbf{A} - \textbf{P}\,\textbf{B}\,\textbf{R}^{-1}\textbf{B}^\top \textbf{P} + \textbf{Q} = \textbf{0},
    \]
    \item $V(t)$ is given by the prediction term
    \[
    V(t) = \mathbb{E}\Bigl[\int_t^\infty \boldsymbol{\Phi}(s,t)^\top \textbf{P}\,d\mathbf{v}(s)\,\Big \vert \,\mathcal{F}_t\Bigr],
    \]
    with the expectation defined over the distribution of the random variable (which is the random future increments of the rough path), $\mathcal{F}_t = \sigma(v(s): 0\leq s \leq t)$ is the natural filtration of the stochastic process $v(t)$, and $\boldsymbol{\Phi}(s,t)$ denoting the fundamental solution of 
    \[
    \frac{d}{ds}\boldsymbol{\Phi}(s,t) = \Bigl(\textbf{A} - \textbf{B}\,\textbf{R}^{-1}\textbf{B}^\top \textbf{P}\Bigr)\boldsymbol{\Phi}(s,t), \quad \boldsymbol{\Phi}(t,t)=\textbf{I}.
    \]
\end{enumerate}
\end{theorem}

\begin{proof}
Using the variation of constants formula~\cite[Chapter~5.3]{book}, we can write
\begin{equation*}
\begin{aligned}
x(t)&=e^{\textbf{A}t}x_0+\int_0^t e^{\textbf{A}(t-s)}\textbf{B}\,u(s)\,ds + \int_0^t e^{\textbf{A}(t-s)}d\mathbf{v}(s). \\
\end{aligned}
\end{equation*}

Define the candidate value function $\Psi(t,x)=0.5x^\top \textbf{P}x+0.5x^\top\textbf{P}V_{noise}+0.5V_{noise}^\top\textbf{P}x+\Upsilon_{noise}$, where $\textbf{P}$, $V_{noise}$, and $\Upsilon_{noise}$ are to be determined. Here, $x^\top \textbf{P}x$ represents the classical quadratic term, $x^\top\textbf{P}V_{noise}$ and $V_{noise}^\top\textbf{P}x$ account for the interaction between the state and the noise correction, and $\Upsilon_{noise}$ represents the cost purely associated with the noise contribution. Compared to the classical LQR problem, the terms $x^\top\textbf{P}V_{noise}$, $V_{noise}^\top\textbf{P}x$, and $\Upsilon_{noise}$ show up in the candidate value function due to the non-semimartingale nature of the process noise.

Define $V_{noise}$ and $\Upsilon_{noise}$ as
\[
V_{noise}(t) = \mathbb{E}\Bigl[\int_t^\infty \boldsymbol{\Phi}(s,t)^\top \textbf{P}\,d\mathbf{v}(s)\,\Big \vert \,\mathcal{F}_t\Bigr],
\]
\begin{equation*}
\begin{aligned}
\Upsilon_{noise}(t)=\int_t^\infty &V_{noise}(s)^\top \textbf{P}\textbf{B}\textbf{R}^{-1}\textbf{B}^\top \textbf{P}V_{noise}(s)ds \\
&- \int_t^\infty V_{noise}(s)^\top \textbf{P}d\textbf{v}(s),
\end{aligned}
\end{equation*}
with $\boldsymbol{\Phi(s,t)}$ denoting the fundamental solution of 
\[
\frac{d}{ds}\boldsymbol{\Phi}(s,t) = \Bigl(\textbf{A} - \textbf{B}\,\textbf{R}^{-1}\textbf{B}^\top \textbf{P}\Bigr)\boldsymbol{\Phi}(s,t), \quad \boldsymbol{\Phi}(t,t)=\textbf{I}.
\]

$V_{\text{noise}}$ is defined to be a predictive correction term, which is a random variable, that estimates the effect of noise on the state trajectory. This means $V_{\text{noise}}$ works on offsetting the expected effect of noise on the optimal state trajectory, given the information up to time $t$.

We define $\tilde{\textbf{v}}$ as the smooth approximation of the rough driver $\textbf{v}$. Based on this approximation, we can define the smooth approximation of the RDE system dynamics as
\[
d\tilde{x} = (\textbf{A}x + \textbf{B}u)\,dt + \tilde{\textbf{v}}\,dt,
\]
\[
\dot{\tilde{x}} = \textbf{A}x + \textbf{B}u + \dot{\tilde{\textbf{v}}}.
\]

Also, we define
\[
\tilde{V}_{noise}(t) = \mathbb{E}\Bigl[\int_t^\infty \boldsymbol{\Phi}(s,t)^\top \textbf{P}\,d\tilde{\mathbf{v}}(s)\,\Big \vert \,\mathcal{F}_t\Bigr],
\]
\begin{equation*}
\begin{aligned}
\tilde{\Upsilon}_{noise}(t)=\int_t^\infty &\tilde{V}_{noise}(s)^\top \textbf{P}\textbf{B}\textbf{R}^{-1}\textbf{B}^\top \textbf{P}\tilde{V}_{noise}(s)ds \\
&- \int_t^\infty \tilde{V}_{noise}(s)^\top \textbf{P}d\tilde{\mathbf{v}}(s),
\end{aligned}
\end{equation*}

Substituting the above smooth approximation into the below HJB equation,
\[
\frac{\partial\Psi^{\tilde{\mathbf{v}}}}{\partial t}+\min_u (\tilde{x}^\top \textbf{Q}\tilde{x}+u^\top \textbf{R}u +\nabla \Psi^{\tilde{\mathbf{v}}} \cdot \dot{\tilde{x}}) = 0,
\]
where 
\[
\frac{\partial \Psi^{\tilde{\mathbf{v}}}}{\partial t} = \tilde{x}^\top \textbf{P}\dot{\tilde{x}} + \tilde{V}_{noise}(t)^\top \textbf{P}\dot{\tilde{x}} + \dot{\tilde{V}}_{noise}(t)^\top \textbf{P}\tilde{x} + \dot{\tilde{\Upsilon}}_{noise}(t),
\]
\[
\nabla \Psi = \textbf{P}\tilde{x} + \textbf{P}\tilde{V}_{noise},
\]
we get that the optimal controller, which ensures that the candidate value function satisfies the HJB equation, is
\[
\tilde{u}^*(t) = -\textbf{R}^{-1}\textbf{B}^\top \textbf{P} \Bigl(\tilde{x}(t) + \tilde{V}(t)\Bigr),
\]
with $\tilde{V}(t)=\tilde{V}_{\text{noise}}(t)$ and $\textbf{P}$ satisfying
\[
\textbf{A}^\top \textbf{P} + \textbf{P}\,\textbf{A} - \textbf{P}\,\textbf{B}\,\textbf{R}^{-1}\textbf{B}^\top \textbf{P} + \textbf{Q} = \mathbf{0}.
\] 

Next, we show that $\tilde{V}$ converges to $V$, which implies that $\Psi^{\tilde{\mathbf{v}}}$ converges to $\Psi$. Set the deterministic integrand
\[
F_t(s):= \boldsymbol{\Phi}(s,t)^\top \textbf{P}.
\]

By the rough path continuity, the rough integral map
\[
\mathbf w \longmapsto \int_t^\infty F_t(s)\,d\mathbf w(s)
\]
is continuous and the integrand is decaying due to Assumption~\ref{ass:mainassumption}. Thus, there exist constants $C>0$ and $\alpha\in(0,1]$ such that, as shown in~\cite{book3},
\[
\Bigg \| \int_t^\infty F_t(s)\,d \tilde{\textbf{v}}(s) - \int_t^\infty F_t(s)\,d \textbf{v}(s)\Bigg \|
\le
C\,\|\widetilde{\mathbf v}-\mathbf v\|_{1/p\text{-H\"ol}}^{\alpha}.
\]

From the deterministic convergence $\|\tilde{\mathbf v}-\mathbf v\|_{1/p\text{-H\"ol}}\to0$, it follows that
\[
\int_t^\infty F_t(s)\,d \tilde{\textbf{v}}(s)\longrightarrow \int_t^\infty F_t(s)\,d \textbf{v}(s).
\]

In the definition of $V(t)$ and $\tilde{V}(t)$, the conditional expectation is with respect to the distribution of the random future increments of the rough path. Since conditional expectation is a bounded linear operator in $L^q$ spaces with respect to $q$-norms, and using dominated convergence, $\tilde{V}(t)$ converges to $V(t)$.

Since the candidate value function and controller satisfy the HJB equation with the smooth approximation of the geometric rough driver, we can take the limit of the smooth approximations to find the solution of the rough HJB equation, as shown in Lemma~\ref{lemma:rough_hjb}. This gives us the optimal controller to be 
\[
u^*(t) = -\textbf{R}^{-1}\textbf{B}^\top \textbf{P} \Bigl(x(t) + V(t)\Bigr),
\]
with $V(t)=V_{\text{noise}}(t)$ and $\textbf{P}$ satisfying
\[
\textbf{A}^\top \textbf{P} + \textbf{P}\,\textbf{A} - \textbf{P}\,\textbf{B}\,\textbf{R}^{-1}\textbf{B}^\top \textbf{P} + \textbf{Q} = \mathbf{0}.
\] 

A completion of squares argument shows that the cost functional can be rewritten as
\begin{equation*}
\begin{aligned}
J(u)= \mathbb{E} \int_0^\infty \| (\textbf{R}^{1/2}&u(t)
+ \textbf{R}^{-1/2}\textbf{B}^\top \textbf{P}x(t) \|^2 dt +\,\mathbb{E} \int_0^\infty 2x(t)^\top \textbf{P} d\textbf{v}(t) \\
\end{aligned}
\end{equation*}

In addition, using a similar completion of squares argument and substituting in the value of the optimal controller, $J(u)$ can be re-written in terms of $J(u^*)$ as
\begin{equation*}
\begin{aligned}
J(u)=J(u^*) &+ \mathbb{E}\int_0^\infty \Bigl[u(t) + \textbf{R}^{-1}\textbf{B}^\top \textbf{P} \bigl(x(t)+V_{\text{noise}}(t)\bigr)\Bigr]^\top \cdot \textbf{R} \Bigl[u(t) + \textbf{R}^{-1}\textbf{B}^\top \textbf{P} \bigl(x(t)+V_{\text{noise}}(t)\bigr)\Bigr]dt.
\end{aligned}
\end{equation*}

Therefore, by the completion of squares argument, we show that for any admissible control $u(t)$, we have
\[
J(u)-J(u^*) = \mathbb{E}\int_0^\infty \Bigl[u(t)-u^*(t)\Bigr]^\top \textbf{R} \Bigl[u(t)-u^*(t)\Bigr]dt \ge 0,
\]
with equality if and only if $u(t)=u^*(t)$ almost everywhere. Hence, the controller
\[
u^*(t) = -\textbf{R}^{-1}\textbf{B}^\top \textbf{P} \Bigl(x(t) + V(t)\Bigr)
\]
is optimal.
\end{proof}

As seen in Theorem~\ref{theorem:glqr}, the optimal controller is made up of two terms. The first term is the same as the optimal controller of the standard LQR problem. The second term is the predictive correction term that estimates and offsets the effect of future noise on the optimal state trajectory, given the information up to time $t$, which is a random variable. It should be noted that if we consider Brownian motion, the predictive correction term is zero, and the optimal controller simplifies to the classical LQR controller.

Due to the use of rough path theory, we can define \emph{pathwise} optimality, which means that the designed controller is optimal for each path, not just in expectation. For any admissible control \(u\), let \(x^{u,\mathbf{v}}\) denote the corresponding solution of (\ref{eq:system_lifted}) driven by a fixed realization \(\mathbf v\), and define the \emph{pathwise} cost as
\begin{equation}
\label{eq:pathwise_cost}
J^{\mathbf v}(u;x_0) = \int_0^\infty \big( x(t)^\top \textbf{Q} x(t) + u(t)^\top \textbf{R} u(t)\big)\,dt.
\end{equation}

Based on the fundamental state transition matrix \(\boldsymbol{\Phi}(s,t)\) for the linear ODE
\begin{equation}
\label{eq:fund_state_transition}
\frac{d}{ds}\boldsymbol{\Phi}(s,t) = (\textbf{A} - \textbf{B}\textbf{R}^{-1}\textbf{B}^\top \textbf{P})\boldsymbol{\Phi}(s,t), \quad \boldsymbol{\Phi}(t,t)=\textbf{I},
\end{equation}
and for the fixed rough path \(\mathbf{v}\), we define the \emph{pathwise} noise correction term as
\begin{equation}
\label{eq:noise_correction}
V_{noise}(t) := \int_t^\infty \boldsymbol{\Phi}(s,t)^\top \textbf{P}\, d\textbf{v}(s),
\end{equation}
which is a well-defined controlled rough path in \(t\).

Corollary~\ref{corollary:pathwise_optimality} formally proves the \emph{pathwise} optimality of the optimal controller in Theorem~\ref{theorem:glqr}. This is a generalization of the pathwise optimality in~\cite{jd}, where standard Brownian motion is considered.

\begin{corollary}[Pathwise Optimality of gLQ Controller]
\label{corollary:pathwise_optimality}
Let \(\textbf{A}\in\mathbb{R}^{n\times n}\), \(\textbf{B}\in\mathbb{R}^{n\times m}\), and let \(\textbf{Q}\in\mathbb{R}^{n\times n}\), \(\textbf{R}\in\mathbb{R}^{m\times m}\) be symmetric positive definite matrices. Let \(x_0\in\mathbb{R}^n\) be given.

Let \(v:[0,\infty)\to\mathbb{R}^n\) be a fixed realization of a continuous path admitting a rough path lift \(\mathbf{v}\) of Hölder regularity \(H\in(\frac{1}{3},1]\).

Consider the controlled rough differential equation
\begin{equation*}
dx(t) = \textbf{A} x(t)\,dt + \textbf{B} u(t)\,dt + d \textbf{v}(t),\qquad x_0 \text{ given},
\end{equation*}
driven by the fixed rough path \(\mathbf v\) and control \(u:[0,\infty)\to\mathbb{R}^m\).

Let $\textbf{P}$ be the unique symmetric positive definite solution of the algebraic Riccati equation
\[
\textbf{A}^\top \textbf{P} + \textbf{P}\,\textbf{A} - \textbf{P}\,\textbf{B}\,\textbf{R}^{-1}\textbf{B}^\top \textbf{P} + \textbf{Q} = \textbf{0}.
\]

Then, for $V_{noise}$ defined in~(\ref{eq:noise_correction}), the control
\begin{equation*}
u^*(t,\mathbf v) \;=\; -\textbf{R}^{-1}\textbf{B}^\top\textbf{P} \big(x^*(t) + V_{noise}(t)\big),
\end{equation*}
where \(x^*\) is the state trajectory controlled with \(u=u^\ast\) that \emph{pathwise} minimizes the cost.
\end{corollary}

\begin{proof}
The proof is \emph{pathwise}: fix a realization of the rough path \(\mathbf{v}\) with Hölder regularity \(H \in (\frac{1}{3},1]\), and work deterministically with controlled rough path theory. All statements apply to almost every sample path of a stochastic process whose rough path lift exists.

By Lyons' universal limit theorem for rough differential equations (Lemma~\ref{lemma:lyonsuniversallimit}), the rough differential equation driven by \(\textbf{v}\) admits a unique controlled rough path solution \(x^{u,\textbf{v}}\). The map \((u,\textbf{v})\mapsto x^{u,\textbf{v}}\) is continuous in the rough path topology.

Under Assumption~\ref{ass:mainassumption}, the algebraic Riccati equation has a unique symmetric positive definite solution \(\textbf{P}\). Thus, the fundamental solution \(\boldsymbol{\Phi}(s,t)\) for~(\ref{eq:fund_state_transition}) exists. Hence, the rough integral $V_{noise}$ in~(\ref{eq:noise_correction}) is well defined for rough path \(\textbf{v}\)~\cite{gubinelli}.

Define the candidate \emph{pathwise} value (cost-to-go) function
\begin{equation*}
\begin{split}
\Psi^{\mathbf v}(x,t) &= 0.5x^\top \textbf{P}x+0.5x^\top\textbf{P}V_{noise}+0.5V_{noise}^\top\textbf{P}x  \\ &+ \int_t^\infty V_{noise}(s)^\top \textbf{P} \textbf{B} \textbf{R}^{-1} \textbf{B}^\top \textbf{P} V_{noise}(s) \, ds \\
&- \int_t^\infty V_{noise}(s)^\top \textbf{P}d\textbf{v}(s).
\end{split}
\end{equation*}

Now, fix an admissible control \(u\) and let \(x(t):=x^{u,\mathbf v}(t)\). For the candidate value \(\Psi^{\mathbf{v}}(x(t),t)\), we compute the \emph{pathwise} differential \(d\big(\Psi^{\mathbf{v}}(x(t),t)\big)\):
\begin{itemize}
  \item \(\nabla_x \Psi ^{\mathbf v}(x,t) = \textbf{P}x + \textbf{P}V_{noise}\).
  \item The time derivative of \(\Psi^{\mathbf v}(x,t)\) equals the sum of \(-x^\top\big(\textbf{A}^\top \textbf{P}+\textbf{P} \textbf{A} - \textbf{P} \textbf{B} \textbf{R}^{-1}\textbf{B}^\top \textbf{P} + \textbf{Q}\big)x\), plus terms coming from the time derivative of \(V_{noise}\) and the explicit integrals in the candidate \emph{pathwise} value function.
  \item The rough integral terms from \(d\textbf{v}\) that appear when differentiating \(\Psi^{\mathbf v}(x,t)\) are exactly canceled by the rough differential of the noise correction \(V_{noise}(t)\), due to the definition of \(V_{noise}(t)\).
\end{itemize}

Collecting all deterministic terms, we obtain the following \emph{pathwise} identity:
\begin{equation*}
\begin{split}
\frac{d}{dt}\Psi^{\mathbf v}(x(t),t) + x(t)^\top \textbf{Q} x(t) + u(t)^\top \textbf{R} u(t) = \big(u(t) + \textbf{R}^{-1}\textbf{B}^\top\textbf{P}(x(t) + V_{noise}(t))\big)^\top \textbf{R} \cdot \big(u(t) + \textbf{R}^{-1}\textbf{B}^\top\textbf{P}(x(t) + V_{noise}(t))\big),
\end{split}
\end{equation*}
valid for every \(t\in[0,\infty)\) and every fixed rough path \(\mathbf v\).
 
Integrating the identity from \(t = 0\) to \(\infty\), we get
\begin{equation*}
\begin{split}
\int_0^\infty (x(s)^\top \textbf{Q} x(s) + u(s)^\top \textbf{R} u(s))\,ds
= \Psi^{\mathbf{v}}(x_0,0) +\int_0^\infty \big\| u(s) - u^*(s,\mathbf{v})\big\|_\textbf{R}^2\,ds,
\end{split}
\end{equation*}
where \(u^*\) is the given optimal controller. Using the definition of \emph{pathwise} cost in~(\ref{eq:pathwise_cost}), we get
\begin{equation*}
\begin{split}
J^{\mathbf v}(u;x_0) = \Psi^{\mathbf{v}}(x_0,0) + \int_0^\infty \big\| u(s) - u^*(s,\mathbf{v})\big\|_\textbf{R}^2\,ds,
\end{split}
\end{equation*}

The quadratic remainder is nonnegative and vanishes if and only if \(u = u^*\). Hence, \(u^*\) minimizes the \emph{pathwise} cost \(J^{\mathbf{v}}\) for this fixed rough path \(\mathbf{v}\).
\end{proof}

\subsection{Optimal Controller without State Observation }
\label{sec:LQG}
Theorem~\ref{theorem:gLQ_correlated} introduces the characterization of the optimal controller in the generalized LQ setting for the correlated noise case, with an observer of the form in \eqref{eq:observer}.

\begin{theorem}[Observer Form of Correlated Noise Case]
\label{theorem:gLQ_correlated}
Consider the problem formulation in Section~\ref{sec:model}. Assume that the state is not directly available but is estimated via the observer
\begin{equation}
\label{eq:observer_sep2}
\begin{aligned}
d\hat{x}(t) &= \textbf{A} \hat{x}(t)dt + \textbf{B} u(t)dt + \textbf{L}\hat{y}(t) dt- \textbf{LC}x(t) dt \\ &~~~~~~~~~~~~~ - \textbf{L}d\textbf{w}(t), \quad \hat{x}_0 \text{ given} \\
\hat{y}(t) &= \textbf{C} \hat{x}(t).
\end{aligned}
\end{equation}

Under Assumption~\ref{ass:mainassumption}, the following hold:
\begin{enumerate}
    \item There exists a unique positive definite solution $\textbf{P}$ to the algebraic Riccati equation
    \[
    \textbf{A}^\top \textbf{P} + \textbf{P}\,\textbf{A} - \textbf{P}\,\textbf{B}\,\textbf{R}^{-1}\textbf{B}^\top \textbf{P} + \textbf{Q} = \textbf{0},
    \]
    \item There exists a unique steady-state solution $\textbf{S}$ to
    \begin{equation*}
    \begin{aligned}
    &\textbf{A}\,\textbf{S}+ \textbf{S}\,\textbf{A}^\top + \Sigma_v - \textbf{S}\textbf{C}^\top\Sigma_w^{-1}\textbf{C}\,\textbf{S} \\ 
    &- \frac{1}{4}(\textbf{R}_{vw}\Sigma_w^{-1}\textbf{R}_{vw}^\top - \textbf{R}_{wv}^\top\Sigma_w^{-1}\textbf{R}_{vw}^\top +3\textbf{R}_{vw}\Sigma_w^{-1}\textbf{R}_{wv}\\
    &+ \textbf{R}_{wv}^\top\Sigma_w^{-1}\textbf{R}_{wv}) - \textbf{R}_{vw}\Sigma_w^{-1}\textbf{C}\textbf{S} - \textbf{S}\textbf{C}^\top\Sigma_w^{-1}\textbf{R}_{wv} = \textbf{0},
    \end{aligned}
    \end{equation*}
    and the optimal observer gain is
    \[
    \textbf{L} = \frac{1}{2}(2\textbf{S}\,\textbf{C}^\top + \textbf{R}_{vw} + \textbf{R}_{wv}^\top)\Sigma_w^{-1},
    \]
    where \( \Sigma_v = \mathbb{E}[d\mathbf{v}(t)d\mathbf{v}(t)^\top]  \), \( \Sigma_w = \mathbb{E}[d\mathbf{w}(t)d\mathbf{w}(t)^\top]  \), \( \textbf{R}_{vw} = \mathbb{E}[d\mathbf{v}(t)d\mathbf{w}(t)^\top] \) and \( \textbf{R}_{wv} = \mathbb{E}[d\mathbf{w}(t)d\mathbf{v}(t)^\top] \).
    \item Let $e(t)=x(t)-\hat{x}(t)$ be the estimation error. Then the cost functional decomposes as
    \[
    J(u)=J_{LQR}(\hat{x},u) + J_{err}(e),
    \]
    with $J_{err}(e) \geq 0$.
    \item The optimal control that minimizes $J(u)$ is given by
    \[
    u^*(t) = -\textbf{R}^{-1}\textbf{B}^\top \textbf{P} \left(\hat{x}(t) + V(t)\right),
    \]
    where
    \[
    V(t)=\mathbb{E}\Bigl[\int_t^\infty \boldsymbol{\Phi}(s,t)^\top \textbf{P}\,d\mathbf{v}(s)\,\Big|\,\mathcal{F}_t\Bigr].
    \]
\end{enumerate}
\end{theorem}

\begin{proof}
The LQR problem with full state information leads to the algebraic Riccati equation, as proven in Theorem~\ref{theorem:glqr},
\[
\textbf{A}^\top \textbf{P} + \textbf{P}\,\textbf{A} - \textbf{P}\,\textbf{B}\,\textbf{R}^{-1}\textbf{B}^\top \textbf{P} + \textbf{Q} = \textbf{0}.
\]

Under the stabilizability and detectability assumptions, it is known that there exists a unique positive definite solution $\textbf{P}$.
  
Since the output is given by $y(t)=\textbf{C}\,x(t)+w(t)$, we implement the observer \eqref{eq:observer_sep2}. Defining the estimation error $e(t)=x(t)-\hat{x}(t)$, and subtracting the observer dynamics from the system dynamics yields
\[
de(t)= (\textbf{A}-\textbf{L}\,\textbf{C})e(t)dt + d\mathbf{v}(t)-\textbf{L}\,d\mathbf{w}(t).
\]

We define the error covariance as
\[
S(t)=\mathbb{E}\left[e(t)e(t)^\top\right].
\]

The time derivative of the error covariance, is 
\begin{equation*}
\begin{aligned}
\frac{d}{dt}S(t)&=\mathbb{E}\left[de(t)e(t)^\top+e(t)de(t)^\top + de(t)de(t)^\top\right] \\
&= \mathbb{E} \big[(\textbf{A}-\textbf{L}\,\textbf{C})e(t)e(t)^\top +e(t)e(t)^\top(\textbf{A}-\textbf{L}\,\textbf{C})^\top \\
&~~~~~~~~~~~~~+ d\mathbf{v}(t) d\mathbf{v}(t)^\top + \textbf{L}d\mathbf{w}(t)d\mathbf{w}(t)^\top \textbf{L}^\top\big] \\
&~~~~~~~~~~~~~- \textbf{L} d\mathbf{w}(t) d\mathbf{v}(t)^\top - d\mathbf{v}(t) d\mathbf{w}(t)^\top \textbf{L}^\top \\
&= (\textbf{A}-\textbf{L}\,\textbf{C})S(t) +S(t)(\textbf{A}-\textbf{L}\,\textbf{C})^\top \\
&~~~~~~~~~~~~~+ \Sigma_v + \textbf{L}\Sigma_w \textbf{L}^\top - \textbf{L}\textbf{R}_{wv} - \textbf{R}_{vw}\textbf{L}^\top,
\end{aligned}
\end{equation*}
under the assumption that $v(t)$ and $w(t)$ are uncorrelated with $e(t)$ for optimal estimation, where \( \Sigma_v = \mathbb{E}[d\mathbf{v}(t)d\mathbf{v}(t)^\top]  \), \( \Sigma_w = \mathbb{E}[d\mathbf{w}(t)d\mathbf{w}(t)^\top]  \), \( \textbf{R}_{vw} = \mathbb{E}[d\mathbf{v}(t)d\mathbf{w}(t)^\top] \) and \( \textbf{R}_{wv} = \mathbb{E}[d\mathbf{w}(t)d\mathbf{v}(t)^\top] \) that are covariance and cross-correlation operators, respectively. They are interpreted as covariance and cross-correlation kernels, or integral operators~\cite{book3,operators}. The interchange between the time derivative and expectation operators is a result of the dominated convergence theorem and the rough Itô formula~\cite{book}. Also, the second order terms in time, i.e. $(dt)^2$, in $\mathbb{E} \big[de(t)de(t)^\top \big]$ can be ignored since they are equal to 0.

For optimal estimation, we aim to minimize the mean squared error. This means
\[
\min\,\mathbb{E} [e(t)^\top e(t)] = \min\,\mathbb{E}[Tr(S(t))],
\]
where $Tr(S(t))$ is the trace of the matrix $S(t)$.

Under detectability and at steady state, there exists a unique stabilizing solution $\textbf{S}$. Finding the partial derivative with respect to $\textbf{L}$ at steady state, we get
\begin{equation*}
\begin{aligned}
\frac{\partial}{\partial L} [(\textbf{A}&-\textbf{L}\,\textbf{C})\textbf{S} +\textbf{S}(\textbf{A}-\textbf{L}\,\textbf{C})^\top + \Sigma_v + \textbf{L}\Sigma_w \textbf{L}^\top \\
& - \textbf{L}\textbf{R}_{wv} - \textbf{R}_{vw}\textbf{L}^\top = \textbf{0}] = -2\textbf{S}\textbf{C}^\top + 2\textbf{L}\Sigma_w \\
&- \textbf{R}_{vw} - \textbf{R}_{wv}^\top =\textbf{0} \\
&\implies \textbf{L}= \frac{1}{2}(2\textbf{S}\,\textbf{C}^\top + \textbf{R}_{vw} + \textbf{R}_{wv}^\top)\Sigma_w^{-1}.
\end{aligned}
\end{equation*}

We let $\textbf{L}= \frac{1}{2}(2\textbf{S}\,\textbf{C}^\top + \textbf{R}_{vw} + \textbf{R}_{wv}^\top)\Sigma_w^{-1}$. Thus, the error covariance can be simplified to
\begin{equation*}
\begin{aligned}
&\textbf{A}\,\textbf{S}+ \textbf{S}\,\textbf{A}^\top + \Sigma_v - \textbf{S}\textbf{C}^\top\Sigma_w^{-1}\textbf{C}\,\textbf{S} \\ 
&- \frac{1}{4}(\textbf{R}_{vw}\Sigma_w^{-1}\textbf{R}_{vw}^\top - \textbf{R}_{wv}^\top\Sigma_w^{-1}\textbf{R}_{vw}^\top +3\textbf{R}_{vw}\Sigma_w^{-1}\textbf{R}_{wv}\\
&+ \textbf{R}_{wv}^\top\Sigma_w^{-1}\textbf{R}_{wv}) - \textbf{R}_{vw}\Sigma_w^{-1}\textbf{C}\textbf{S} - \textbf{S}\textbf{C}^\top\Sigma_w^{-1}\textbf{R}_{wv} = \textbf{0}.
\end{aligned}
\end{equation*}

This choice of observer gain $\textbf{L}$ ensures that the estimation error is orthogonal to the data used to generate the estimate, which means that it minimizes the mean squared error. 

In summary, there exists a unique stabilizing solution $\textbf{S}$, and the optimal observer gain is chosen as
\[
\textbf{L} = \frac{1}{2}(2\textbf{S}\,\textbf{C}^\top + \textbf{R}_{vw} + \textbf{R}_{wv}^\top)\Sigma_w^{-1},
\]
given that
\begin{equation*}
\begin{aligned}
&\textbf{A}\,\textbf{S}+ \textbf{S}\,\textbf{A}^\top + \Sigma_v - \textbf{S}\textbf{C}^\top\Sigma_w^{-1}\textbf{C}\,\textbf{S} \\ 
&- \frac{1}{4}(\textbf{R}_{vw}\Sigma_w^{-1}\textbf{R}_{vw}^\top - \textbf{R}_{wv}^\top\Sigma_w^{-1}\textbf{R}_{vw}^\top +3\textbf{R}_{vw}\Sigma_w^{-1}\textbf{R}_{wv}\\
&+ \textbf{R}_{wv}^\top\Sigma_w^{-1}\textbf{R}_{wv}) - \textbf{R}_{vw}\Sigma_w^{-1}\textbf{C}\textbf{S} - \textbf{S}\textbf{C}^\top\Sigma_w^{-1}\textbf{R}_{wv} = \textbf{0}.
\end{aligned}
\end{equation*}

Writing $x(t)=\hat{x}(t)+e(t)$ and substituting into the cost functional yields
\[
J(u)=\mathbb{E}\int_0^\infty \Bigl[(\hat{x}(t)+e(t))^\top \textbf{Q} (\hat{x}(t)+e(t)) + u(t)^\top \textbf{R}\,u(t)\Bigr]dt.
\]

Through algebraic manipulations and the properties of the controlled error $e(t)$ shows that the cost decomposes as
\begin{equation*}
\begin{aligned}
J(u)&= \mathbb{E}\int_0^\infty \Bigl[\hat{x}(t)^\top \textbf{Q} \hat{x}(t) + u(t)^\top \textbf{R}\,u(t)\Bigr]dt + \mathbb{E}\int_0^\infty \Bigl[ e(t)^\top \textbf{Q} e(t) \Bigr]dt\\
&=J_{LQR}(\hat{x},u) + J_{err}(e),
\end{aligned}
\end{equation*}
with $J_{err}(e)\geq 0$, and equality only occurs if and only if $e(t) = 0$. Therefore, the optimal control minimizing $J(u)$ is obtained by solving the LQR problem with state $\hat{x}(t)$:
\[
u^*(t) = -\textbf{R}^{-1}\textbf{B}^\top\textbf{P} \left(\hat{x}(t) + V(t)\right).
\]

The cost decomposition implies that the design of the state feedback controller and the observer are decoupled. Thus, the separation principle holds and the overall optimal controller is given by
\[
u^*(t) = -\textbf{R}^{-1}\textbf{B}^\top \textbf{P}\left(\hat{x}(t) + V(t)\right).
\]
\end{proof}

By setting $\textbf{R}_{wv} = \textbf{0}$ and $\textbf{R}_{vw} = \textbf{0}$ in Theorem~\ref{theorem:gLQ_correlated}, we can get the characterization of the optimal controller in the generalized LQ setting for the uncorrelated noise case, with an observer of the form in~\eqref{eq:observer}. This shows the optimal controller is the same as the optimal controller of the LQR case with the exception of the use of the state estimation of the observer, rather than the use of the true state. Also, the case of correlated process and measurement noises in Theorem~\ref{theorem:gLQ_correlated} is shown to include more terms due to the consideration of the information from the correlation of the noises.

Furthermore, it is known that, in general, Itô solution maps are not continuous~\cite{book2}. On the contrary, the rough path framework provides solution maps that are locally uniformly continuous~\cite{book2}. This provides another motivation for the use of rough path theory over classical stochastic calculus for the case when both solutions are well defined. Theorem~\ref{theorem:continuity} formally states this continuity result.

\begin{theorem}
\label{theorem:continuity}
Consider the below system dynamics, state estimation filter, and controller, where the noises $dv(t)$ and $dw(t)$ are assumed to be correlated semimartingales:
\begin{equation*}
\begin{aligned}
dx(t) &= \textbf{A} x(t) dt + \textbf{B} u(t) dt + dv(t), \\
d\hat{x}(t) &= \textbf{A} \hat{x}(t)dt + \textbf{B} u(t)dt + \textbf{L} (\hat{y}(t)- y(t)) dt - \textbf{L}dw(t), \\
y(t) &= \textbf{C} x(t), \\
\hat{y}(t) &= \textbf{C} \hat{x}(t), \\
u^*(t) &= -\textbf{R}^{-1}\textbf{B}^\top \textbf{P} \left(\hat{x}(t) + V(t)\right).
\end{aligned}
\end{equation*}

For the above system, the following conclusions can be deduced:

\begin{itemize}
    \item \textbf{Raw Itô continuity (supremum norm on paths)} \\
    The maps 
    \[
    y \;\mapsto\; \hat{x}, \text{  and  } y \;\mapsto\; u^*,
    \]
    are, in general, not continuous in supremum norm nor in any other metric on path space, since they depend on stochastic integrals whose values can jump under noise perturbations~\cite{book2,lyons,lyons2,book3}.
    
    \item \textbf{Rough path continuity (enhanced observation)} \\
    We lift the state estimation filter to its rough path and equip the path space with the $\alpha$-Hölder rough path metric for $\alpha \in (\frac{1}{3},1]$. In this topology, there exists a filtering solution map that is locally uniformly continuous such that the solutions from the raw Itô map and the enhanced observation map are equal almost surely in probability. 
\end{itemize}
\end{theorem}

\begin{proof}
The proof of the discontinuity of the Itô solution follows from Theorem 1.1.1 in~\cite{book2}.

The proof of rough path continuity is a composition of three continuity statements.

\textbf{1) Continuity of the estimator map} \\
By~\cite{continuity}, the filter solution mapping that associates the solution of the stochastic differential equation to an enhanced observation path is locally uniformly continuous. Concretely, Theorem 2 (existence of solution) and Theorem 3 (continuity) in~\cite{continuity} provide the detailed proof for the continuity of the mapping, which is the rough path replacement of the Clark robust representation in~\cite{clark}. This means that the solution of 
\[
d\hat{x}(t) = \textbf{A} \hat{x}(t)dt + \textbf{L} \hat{y}(t) dt - \textbf{L}d\textbf{w}(t)
\]
is continuous.
In addition, the dependence on the state can be viewed, in this case, as an input. Thus, if the state and controller input are continuous, by Lyons’ universal limit theorem (Lemma~\ref{lemma:lyonsuniversallimit}) and the standard continuity for rough differential equations with a time-dependent continuous forcing term imply that the solution map for the estimator is still locally uniformly continuous, despite the dependence on the state and input under the assumption that they are continuous. This means that the solution of 
\begin{equation*}
\begin{split}
d\hat{x}(t) &= \textbf{A} \hat{x}(t)dt + (\textbf{B} u(t) - \textbf{L} y(t))dt + \textbf{L} \hat{y}(t) dt - \textbf{L}d\textbf{w}(t) \\
d\hat{x}(t) &= \textbf{A} \hat{x}(t)dt + \eta(u(t), y(t)) dt + \textbf{L} \hat{y}(t) dt - \textbf{L}d\textbf{w}(t) 
\end{split}
\end{equation*}
is continuous, under the assumption that $\eta(u(t), y(t))$ is continuous.

\textbf{2) Continuity of the controller as a map of the estimator} \\
The controller we use is linear in the state estimator. $V(t)$ is the predictive correction term built from the process noise, which is a fixed rough path in our set-up. For fixed $\textbf{v}$, $V(t)$ is a fixed controlled rough integral, hence, it is fixed. Therefore, the map of the enhanced observation to the control input is a composition of a linear operator on the estimator map and the addition of a fixed term. Thus, this map is locally uniformly continuous, which follows directly from linearity and the continuity from step 1. This means that 
\[
u^*(t) = -\textbf{R}^{-1}\textbf{B}^\top \textbf{P} \left(\hat{x}(t) + V(t)\right)
\]
is continuous.

\textbf{3) Continuity of the solution of the state rough differential equation solution} \\
For a fixed process noise sample path, its rough path lift, $\textbf{v}$, is fixed. The state trajectory solves the controlled rough differential equation. Lyons’ universal limit theorem (Lemma~\ref{lemma:lyonsuniversallimit}) and the standard continuity for rough differential equations with a time-dependent continuous forcing term imply that the solution map for the state trajectory is locally uniformly continuous in the appropriate rough path topology. Thus, combining with step 2, the continuity of the controller gives continuity of the composition: from the enhanced observation to the control input to the state trajectory. This means that the solution of
\[
dx(t) = \textbf{A} x(t) dt + \textbf{B} u(t) dt + d\textbf{v}(t)
\]
is continuous.

Putting the three steps together gives the desired result, that is the state trajectory is continuous for a fixed rough sample path. Thus, the proof is as follows:
\begin{align*}
\begin{array}{c}
\text{rough} \\
\text{path lift}
\end{array}
\xrightarrow[\text{Lyons' universal limit theorem}]{\text{Theorems 2 and 3 in~\cite{continuity}}}
\hat{x} 
\xrightarrow[\text{map}]{\text{linear}}
u^*
\xrightarrow[\text{limit theorem}]{\text{Lyons' universal}}
x.
\end{align*}

Therefore, the lifted maps yield stable continuous limits for $\hat{x}$ and $u^*$. This leads to the state trajectory map, which is the solution of the rough differential equation driven by fixed $\textbf{v}$ and control $u^*$, being locally uniformly continuous. In contrast, raw Itô maps without L\'evy area information fail to converge in supremum norm or in any other metric on the path space.
\end{proof}

\section{Illustrative Examples}
\label{sec:example}
In this section, we provide two illustrative examples to demonstrate the feasibility of the control strategies discussed earlier, as well as evaluating the performance of the classical LQ controller against the proposed gLQ controller, particularly when the system is subjected to non-Markovian and non-semimartingale stochastic processes. The first illustrative example subjects the system to continuous fBm noise, while the second illustrative examples subjects the system to a L\'evy jump process noise that follows a Stable distribution.

We consider a linear inverted pendulum model with additive non-Markovian and non-semimartingale noise with a quadratic cost function. Table I provides the parameters of the inverted pendulum system.

\begin{table}[!htbp]
\label{table:params}
\caption{Parameters of inverted pendulum system}
\centering
\begin{tabular}{|c|c|c|}
\hline
\textbf{Symbol} & \textbf{Description} & \textbf{Value} \\
\hline
$M$ & Mass of the pendulum rod & $0.1 \text{ kg}$ \\
\hline
$M_c$ & Mass of the cart & $0.135 \text{ kg}$ \\
\hline
$L$ & Pendulum length from pivot & $0.2 \text{ m}$ \\
&  to center of gravity & \\
\hline
$J_m$ & Motor rotor moment of inertia & $3.26 \times 10^{-8} \text{ kg} \cdot \text{m}^2$ \\
\hline
$R_m$ & Motor armature resistance & $12.5 \, \Omega$ \\
\hline
$k_b$ & Motor back emf constant & $0.031 \text{ V/rad/sec}$ \\
\hline
$k_t$ & Motor torque constant & $0.031 \text{ N}\cdot \text{m/A}$ \\
\hline
$R$ & Motor pinion radius & $0.006 \text{ m}$ \\
\hline
$B$ & Viscous damping at pivot & $0.000078 \text{ N} \cdot \text{m/rad/sec}$ \\
& of pendulum & \\
\hline
$C$ & Viscous friction coefficient & $0.63 \text{ N/m/sec}$ \\
& for cart displacement & \\
\hline
$I$ & Mass moment of inertia & $0.00072 \text{ kg} \cdot \text{m}^2$ \\
& of pendulum rod & \\
\hline
$M$ & Total cart weight mass & $0.136 \text{ kg}$ \\
& including motor inertia & \\
\hline
$G$ & Gravitational constant & $9.81 \text{ m/sec}^2$ \\
\hline
\end{tabular}
\end{table}

We consider the state vector to be 

\[
x =
\begin{bmatrix}
d \\ \theta \\ \dot{d} \\ \dot{\theta}
\end{bmatrix},
\]

where $d$ is the linear displacement and $\theta$ is the angular displacement.

Thus, the system matrices are given below.

\[
\textbf{A} =
\begin{bmatrix}
0 & 0 & 1 & 0 \\
0 & 0 & 0 & 1 \\
0 & \frac{m^2 l^2 g}{\alpha} &
    -\frac{(I + m l^2)}{\alpha} \left(c + \frac{k_t k_b}{R_m r^2}\right) &
    -\frac{b m l}{\alpha} \\
0 & \frac{m g l (M + m)}{\alpha} &
    -\frac{m l}{\alpha} \left(c + \frac{k_t k_b}{R_m r^2}\right) &
    -\frac{b(M+m)}{\alpha}
\end{bmatrix}.
\]

\[
\textbf{B} =
\begin{bmatrix}
0 \\
0 \\
\frac{(I + m l^2) k_t}{\alpha R_m r} \\
\frac{m l k_t}{\alpha R_m r}
\end{bmatrix}.
\]

Plugging in the system parameters to the system matrices, we get the following:
\[
\textbf{A} =
\begin{bmatrix}
0 & 0 & 1 & 0 \\
0 & 0 & 0 & 1 \\
0 & 5.51 & -18.29 & -0.002 \\
0 & 64.9 & -77.53 & -0.026
\end{bmatrix},
\]

\[
\textbf{B} =
\begin{bmatrix}
0 \\
0 \\
2.73 \\
11.59
\end{bmatrix}, \quad
\textbf{C} =
\begin{bmatrix}
1 & 0.1 & 0.1 & 0.1 \\
0.1 & 1 & 0.1 & 0.1 \\
0.1 & 0.1 & 1 & 0.1 \\
0.1 & 0.1 & 0.1 & 1
\end{bmatrix}.
\]

During simulations, the controller output is inputted through a saturation function to saturate the value at -1000 and +1000 to avoid large controller outputs.

\subsection{Fractional Brownian Motion}
In this example, the system is subjected to fBm noise with Hölder regularity $H = 0.35$.

As depicted in Fig.~\ref{fig:LQ-fBm}, the classical LQ controller demonstrates a significant inability to maintain system stability when subjected to fBm noise. The state trajectories, encompassing position, angle, and their respective rates of change, exhibit an uncontrolled and rapid divergence from their equilibrium points. Specifically, after approximately 4 seconds into the simulation, the magnitudes of all components of the state vector diverge exponentially. This exponential growth across all state variables unequivocally indicates that the classical LQ controller is not robust to fBm noise, leading to an unstable and unbounded system response.

Further, the controller output is characterized by an immediate and drastic negative deflection, plummets to its saturation point within the first second of the simulation. This large, sustained negative control effort signifies the controller's futile attempt to counteract the destabilizing effects of the fBm noise. The combination of unbounded state trajectories and a persistently saturated control output highlights the fundamental limitations of the classical LQ framework when applied to systems influenced by non-Markovian noise processes such as fBm.

\begin{figure} 
\centering
\includegraphics[width=\linewidth]{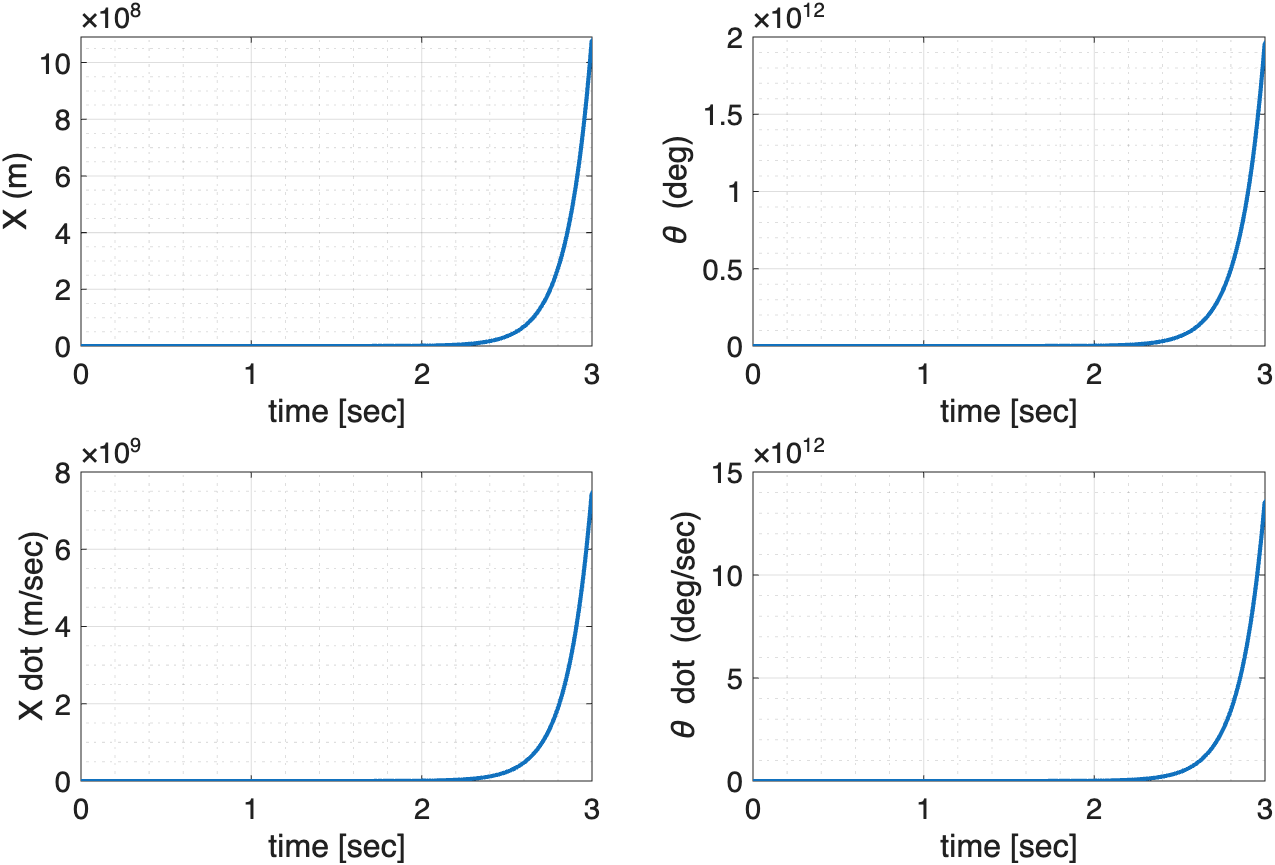}
\caption{State trajectory under fBm noises with classical LQ controller}
\label{fig:LQ-fBm}
\end{figure}

In stark contrast to the classical LQ controller, the gLQ controller demonstrates remarkable efficacy in stabilizing the system under fBm noise, as evidenced by the state trajectories presented in Fig.~\ref{fig:gLQ-fBm}. All four state variables are effectively regulated and converge to stable, bounded values within the simulation period. The angle stabilizes around 180 degrees, as expected, while the rate of change of the angle converges to zero, demonstrating effective control over the system. This significant improvement in stability underscores the robustness of the gLQ controller, which leverages rough path theory to effectively handle the complexities introduced by fBm noise.

The control output generated by the rough path gLQ controller under fBm noise further reinforces its superior performance. The controller output exhibits an initial transient response, followed by a rapid decay and subsequent stabilization near zero, with only minor fluctuations. This behavior is fundamentally different from the persistently large control effort observed with the classical LQ controller. The gLQ controller's ability to achieve system stability with a significantly smaller and more controlled output highlights its efficiency.

\begin{figure} 
\centering
\includegraphics[width=\linewidth]{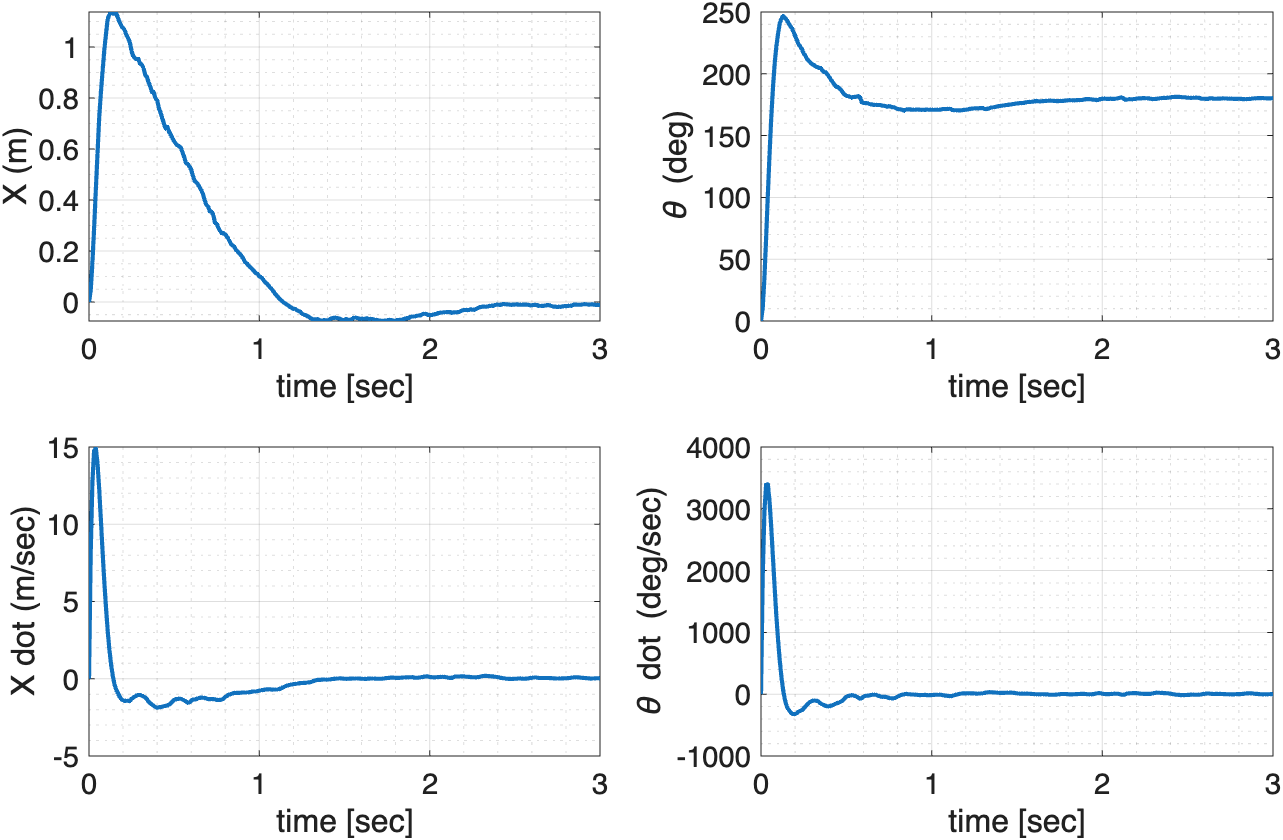}
\caption{State trajectory under fBm noises with rough path gLQ controller}
\label{fig:gLQ-fBm}
\end{figure}

\subsection{L\'evy Processes (Stable Distribution)}
In this example, the system is subjected to a jump process noise that follows a Stable distribution with $\alpha = 1.5$.

The performance of the classical LQ controller under Stable distribution noise mirrors its failure under fBm noise, as illustrated in Fig.~\ref{fig:LQ-Stable}. The state trajectories exhibit a similar pattern of uncontrolled divergence, rapidly escalating to extremely large magnitudes. This pervasive instability across all state variables confirms that the classical LQ controller is fundamentally ill-equipped to handle the heavy-tailed characteristics of Stable distribution noise.

The control effort exhibits an immediate and sustained large saturated value, similar to the fBm case. This persistent and extreme control signal, coupled with the unbounded state trajectories, underscores the classical LQ controller's inability to effectively mitigate the disruptive effects of Stable distribution noise.

\begin{figure} 
\centering
\includegraphics[width=\linewidth]{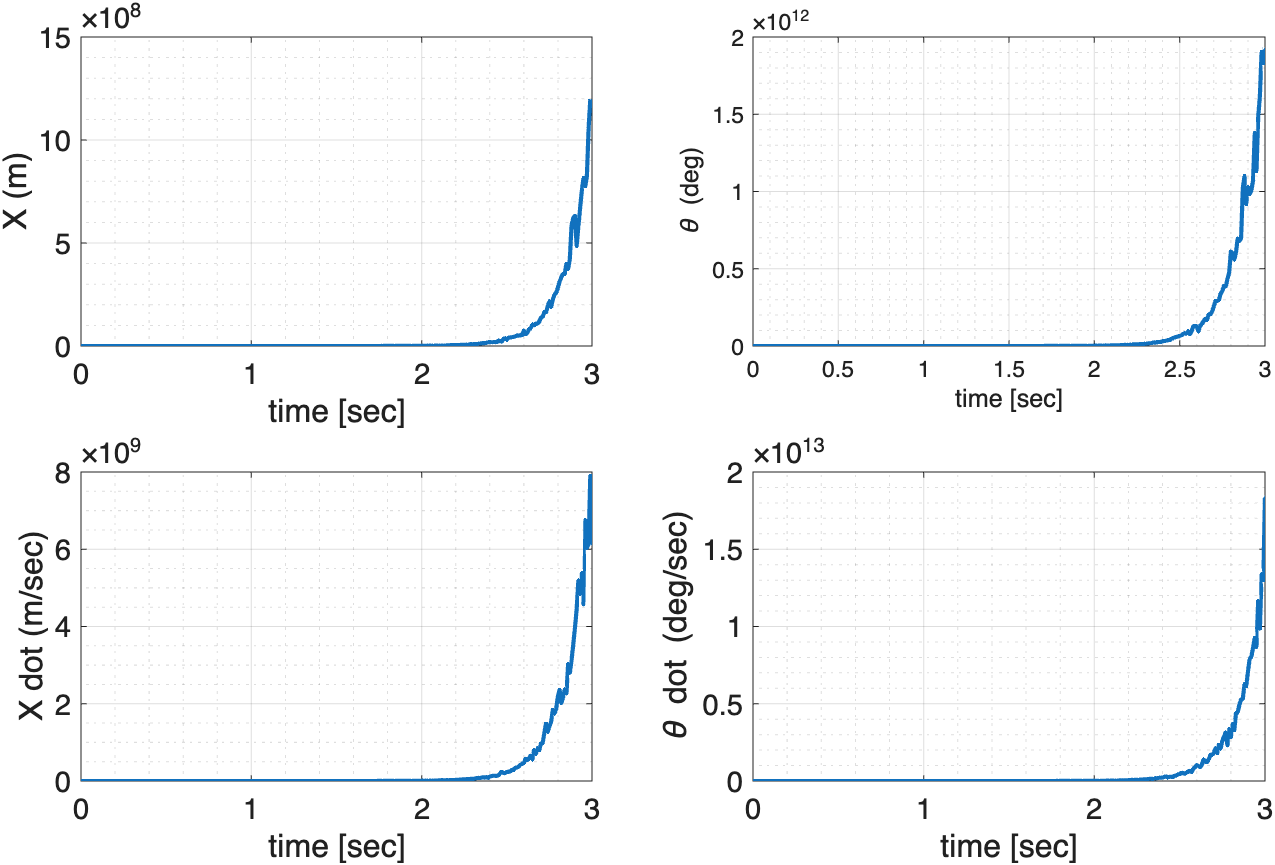}
\caption{State trajectory under Stable distribution noises with classical LQ controller}
\label{fig:LQ-Stable}
\end{figure}

Consistent with its performance under fBm noise, the rough path gLQ controller demonstrates robustness and effectiveness in stabilizing the system under Stable distribution noise, as shown in Fig.~\ref{fig:gLQ-Stable}. All state variables are successfully regulated and converge to bounded values. The angle stabilizes at around 180 degrees, as expected, and its rate of change converges to around zero. This remarkable ability to maintain stability under such challenging noise conditions highlights the significant advantages of incorporating rough path theory into the LQ framework, enabling the gLQ controller to effectively manage systems influenced by heavy-tailed noise characteristics that render classical methods ineffective.

Similar to its behavior under fBm noise, the control output exhibits an initial transient response, followed by a rapid convergence to a value near zero, with minimal fluctuations. This controlled and bounded output, in conjunction with the stable state trajectories observed in Fig.~\ref{fig:gLQ-Stable}, further emphasizes the performance and efficiency of the gLQ controller. It demonstrates that the gLQ approach can achieve system stability with a reasonable and well-behaved control effort, even when confronted with highly disruptive Stable distribution noise.

\begin{figure} 
\centering
\includegraphics[width=\linewidth]{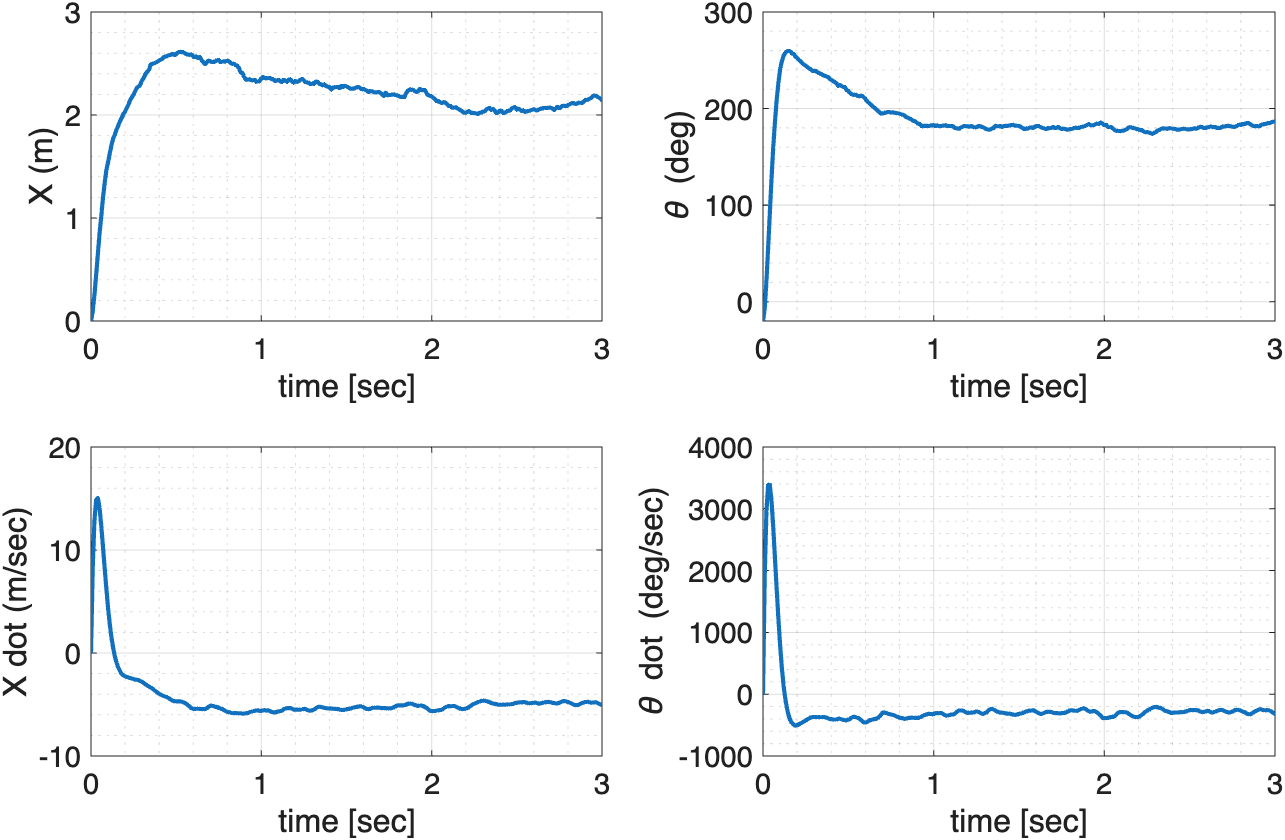}
\caption{State trajectory under Stable distribution noises with rough path gLQ controller}
\label{fig:gLQ-Stable}
\end{figure}

Thus, the results unequivocally demonstrate the profound impact of noise characteristics on controller performance and highlight the significant advantages of the gLQ controller, which incorporates rough path theory, over the classical LQ controller. The classical LQ framework, traditionally effective for systems perturbed by Gaussian and Markovian noise, proved to be fundamentally inadequate when faced with fBm noise and Stable distribution noise. In both scenarios, the classical LQ controller consistently failed to stabilize the system, leading to unbounded state trajectories and persistently large, ineffective control efforts. This failure can be attributed to the classical LQ controller's reliance on assumptions of white Gaussian noise, which are violated by fBm noise and Stable distribution noise. Also, the performance of the gLQ controller is a direct consequence of the gLQ framework's ability to account for the non-Markovian and non-semimartingale nature of the noise processes through the rigorous mathematical tools of rough path theory.

\section{Conclusion}
\label{sec:conc}
In this paper, we have extended the classical LQ control framework to accommodate non-Markovian and non-semimartingale noise models with low Hölder regularity in both process and measurement noises, leading to the formulation of the gLQ problem. Recognizing the limitations of traditional LQ methods in handling dependencies and rough trajectories, we employed rough path theory to develop a well-posed approach for optimal state estimation and control. By leveraging the iterated integral and controlled rough paths, we explicitly characterized the optimal controller for systems influenced by non-Markovian and non-semimartingale noise models with Hölder regularity $H \in (1/3,1]$. The numerical simulations of the inverted pendulum model validate our approach, demonstrating its efficacy in mitigating the challenges introduced by non-Markovian and non-semimartingale noise models in control systems, particularly fBm and L\'evy processes.

Future research directions can encompass the utilization of rough path theory for the use of path signatures in state representation for reinforcement learning agents to account for long-term dependencies and path-dependent behaviors. This opens the field to having path-sensitive policies, path-dependent rewards, and learning representations of rough time series. This could potentially increase the usability of multi-agent reinforcement learning-based control methods in physical systems, as a result of more generalized system modeling, and more accurate modeling of disturbances from natural phenomena. Moreover, the ability to have robust controllers against more general disturbance models can lead to better security and safety guarantees in multi-agent reinforcement learning and distributed control.

\appendices
\section{Supporting Theorems and Lemmas}
\label{app1}
\begin{lemma}[Sewing Lemma in~\cite{book}]
\label{lemma:sewing}
Let $\alpha$ and $\beta$ be such that $0 < \alpha \leq 1 < \beta$. Then, there exists a unique continuous linear map $\mathcal{I} : C^{\alpha,\beta}_2 ([0, T], W) \to C^\alpha([0, T], W)$ such that $(\mathcal{I}\Xi)_0 = 0$ and
\[
\| (\mathcal{I}\Xi)_{s,t} - \Xi_{s,t} \| \leq C \vert t - s \vert^{\beta},
\]
where $C$ only depends on $\beta$ and $\|\delta \Xi\|_{\beta}$ (The $\alpha$-Hölder norm of $\mathcal{I}\Xi$ also depends on $\|\Xi\|_{\alpha}$ and hence on $\|\Xi\|_{\alpha,\beta}$).
\end{lemma}

\begin{lemma}[Lyons' Universal Limit Theorem in~\cite{book2}]
\label{lemma:lyonsuniversallimit}
Let $p \geq 1$ and suppose that $\gamma > p$. Assume that 
\[
f: \mathbb{R}^e \to L(\mathbb{R}^d, \mathbb{R}^e)
\]
is a $\gamma$-Lipschitz function (in the sense that there exists a constant $K>0$ such that for all $x,y\in \mathbb{R}^e$, 
\[
\| f(x) - f(y) \| \le K \| x-y \|^\gamma,
\]
and let $a \in \mathbb{R}^e$ be the initial condition.

Then, for any $p$-rough path $\mathbf{X} = (X,\mathbb{X})$ defined on the interval $[0,T]$, the rough differential equation
\begin{equation}
\label{eq:RDE}
dy_t = f(y_t)\, d\mathbf{X}_t, \quad y_0 = a,
\end{equation}
admits a unique solution $y \in C^{p}([0,T];\mathbb{R}^e)$.

Moreover, the solution map
\[
\mathcal{I} : \mathbf{X} \mapsto y,
\]
which assigns to each $p$-rough path $\mathbf{X}$ the unique solution $y$ of \eqref{eq:RDE}, is locally Lipschitz continuous with respect to the $p$-variation metric.
\end{lemma}

\section{Preliminaries}
\label{app2}

\subsection{Fractional Brownian Motion}
Fractional Brownian motion (fBm) is a family of Gaussian processes that is characterized by the Hurst parameter $H \in (0,1)$.

\begin{definition}
Let $H \in (0,1)$ be fixed. A real-valued standard fractional Brownian motion, $(B(t), t \geq 0)$, with Hurst parameter $H$ is a Gaussian process with continuous sample paths such that
\begin{align}
    \mathbb{E}[B(t)] &= 0, \\
    \mathbb{E}[B(s)B(t)] &= \frac{1}{2} \left( t^{2H} + s^{2H} - |t- s|^{2H} \right),
\end{align}
for all $s, t \in \mathbb{R}^{+}$.
An $\mathbb{R}^n$-valued standard fractional Brownian motion, $(B(t), t \geq 0)$, with Hurst parameter $H$ is an $n$-vector of independent real-valued standard fractional Brownian motions with the same Hurst parameter $H$. If $B$ is an fBm with $H = 1/2$, then $B$ is a Brownian motion.
\end{definition}

If $(B(t), t \geq 0)$ is a real-valued standard fractional Brownian motion, then for each $\alpha > 0$, $(B(\alpha t), t \geq 0)$ and $(\alpha^H B(t), t \geq 0)$ have the same probability law. This property is called self-similarity. For $H \in (1/2,1)$, we have the long-range dependence property:
\begin{equation}
    \sum_{n=1}^{\infty} r(n) = \infty,
\end{equation}
where
\begin{equation}
    r(n) = \mathbb{E}[B(1)(B(n+1) - B(n))].
\end{equation}

Fractional calculus plays an important role in the analysis of an fBm. Let $\alpha \in (0,1)$ be fixed. The left-sided and the right-sided fractional integrals for $\varphi \in L^1([0,T])$ are defined for almost all $t \in [0,T)$ by
\begin{align}
    (I^{\alpha}_{0+} \varphi)(t) &= \frac{1}{\Gamma (\alpha)} \int_{0}^{t} (t - s)^{\alpha-1} \varphi(s)ds, \\
    (I^{\alpha}_{T-} \varphi)(t) &= \frac{(-1)^{-\alpha}}{\Gamma (\alpha)} \int_{t}^{T} (s - t)^{\alpha-1} \varphi(s)ds,
\end{align}
respectively, where $\Gamma (\cdot)$ is the gamma function. 

One indication of the relevance of fractional calculus to fBm is that the covariance of an fBm can be expressed in terms of fractional integrals. Let $(B(t), t \geq 0)$ be a real-valued standard fractional Brownian motion. The covariance of $B$ satisfies the following equality:
\begin{equation}
\begin{aligned}
    \mathbb{E}[B(s)B(t)] = \rho(H) \int_{0}^{T} u^2_{1/2 - H}(r) (I^{H-1/2}_{T-} u_{H-1/2} 1_{[0,s]})(r) \cdot (I^{H-1/2}_{T-} u_{H-1/2} 1_{[0,t]})(r)dr,
\end{aligned}
\end{equation}
where
\begin{equation}
    \rho(H) = \frac{2H\Gamma(H + 1/2)\Gamma(3/2 - H)}{\Gamma(2 - 2H)}.
\end{equation}
and $u_a(s) = s^a \text{ for } a > 0, s > 0$.

\subsection{Stable Distribution}
Stable distribution is a class of probability distributions suitable for modeling skewness and heavy tails. A linear combination of two independent, identically-distributed stable-distributed random variables has the same distribution as the individual variables. This means, if $X_1, X_2, \ldots, X_n$ are independent and identically distributed stable random variables, then for every $n$, there exists constants $c_n > 0$ and $d_n \in \mathbb{R}$ such that
\[
X_1 + X_2 + \cdots + X_n = c_n X + d_n.
\]

The stable distribution is a direct application of the generalized central limit theorem, which states that the limit of normalized sums of independent identically distributed variables is stable.

There are a number of different parameterizations for the stable distribution in the literature. For example, a random variable $X$ has a stable distribution $S(\alpha, \beta, \gamma, \delta)$ if its characteristic function is given by:
\[
\mathbb{E}(e^{itX}) = 
\begin{cases}
\exp\!\Bigl( -\gamma^\alpha |t|^\alpha \Bigl[1 + i \beta \, \text{sign}(t) 
\tan \frac{\pi \alpha}{2} \\ ~~~~~~~~~~~ \cdot \bigl((\gamma |t|)^{1-\alpha} - 1 \bigr) \Bigr] + i \delta t \Bigr), & \text{ for } \alpha \neq 1, \\[1em]
\exp\!\Bigl( -\gamma |t| \Bigl[ 1 + i \beta \, \text{sign}(t) \frac{2}{\pi} \ln (\gamma |t|) \Bigr] \\ ~~~~~~~~~~~~+ i \delta t \Bigr), & \text{ for } \alpha = 1.
\end{cases}
\]

Thus, the stable distribution is defined using the following parameters:

\begin{center}
\begin{tabular}{llll}
\toprule
\textbf{Parameter} & \textbf{Description} & \textbf{Support} \\
\midrule
alpha ($\alpha$) & First shape parameter  & $0 < \alpha \leq 2$ \\
beta ($\beta$)   & Second shape parameter & $-1 \leq \beta \leq 1$ \\
gam ($\gamma$)   & Scale parameter        & $0 < \gamma < \infty$ \\
delta ($\delta$) & Location parameter     & $-\infty < \delta < \infty$ \\
\bottomrule
\end{tabular}
\end{center}

The first shape parameter, $\alpha$, describes the tails of the distribution. For $\alpha \neq 2$, the stable distribution defines a jump process.

The second shape parameter, $\beta$, describes the skewness of the distribution. If $\beta = 0$, then the distribution is symmetric. If $\beta > 0$, then the distribution is right-skewed. If $\beta < 0$, then the distribution is left-skewed. When $\alpha$ is small, the skewness of $\beta$ is significant. As $\alpha$ increases, the effect of $\beta$ decreases.



\bibliographystyle{IEEEtran}
\bibliography{refs}

@book{book,
author = {Peter K. Friz and Martin Hairer},
title = {A Course on Rough Paths: With an Introduction to Regularity Structures},
year = {2020},
publisher = {Springer Cham},
series = {Universitext},
doi = {10.1007/978-3-030-41556-3}}

@book{book2,
author = {Lyons, Terry and Qian, Zhongmin},
title = {System Control and Rough Paths},
publisher = {Oxford University Press},
year = {2002},
month = {12},
isbn = {9780198506485},
doi = {10.1093/acprof:oso/9780198506485.001.0001},
url = {https://doi.org/10.1093/acprof:oso/9780198506485.001.0001}}

@book{book3,
place = {Cambridge},
series = {Cambridge Studies in Advanced Mathematics},
title = {Multidimensional Stochastic Processes as Rough Paths: Theory and Applications},
publisher = {Cambridge University Press},
author = {Friz, Peter K. and Victoir, Nicolas B.},
year = {2010},
collection = {Cambridge Studies in Advanced Mathematics}}

@article{qz,
title = {Recent Advances in Non-Gaussian Stochastic Systems Control Theory and Its Applications}, volume={1}, url = {https://www.sciltp.com/journals/ijndi/article/view/169}, doi = {10.53941/ijndi0101010},
number = {1},
journal = {International Journal of Network Dynamics and Intelligence},
author = {Zhang, Qichun and Zhou, Yuyang},
year = {2022},
month = {12},
pages = {111–119}}

@INPROCEEDINGS{bg,
author = {Goodwine, Bill},
booktitle = {2023 31st Mediterranean Conference on Control and Automation (MED)}, 
title = {Fractional-Order Dynamics in Large Scale Control Systems}, 
year = {2023},
pages = {747-752},
doi = {10.1109/MED59994.2023.10185897}}

@article{ted,
title = {DISTRIBUTED PARAMETER SYSTEMS WITH A MULTIPLICATIVE FRACTIONAL GAUSSIAN NOISE},
journal = {IFAC Proceedings Volumes},
volume = {38},
number = {1},
pages = {35-38},
year = {2005},
note = {16th IFAC World Congress},
issn = {1474-6670},
doi = {https://doi.org/10.3182/20050703-6-CZ-1902.00866},
url = {https://www.sciencedirect.com/science/article/pii/S1474667016368781},
author = {T.E. Duncan and B. Pasik-Duncan}}

@article{lyons,
author = {Lyons, Terry J.},
journal = {Revista Matemática Iberoamericana},
number = {2},
pages = {215-310},
title = {Differential equations driven by rough signals},
url = {http://eudml.org/doc/39555},
volume = {14},
year = {1998}}

@article{gubinelli,
title = {Controlling rough paths},
journal = {Journal of Functional Analysis},
volume = {216},
number = {1},
pages = {86-140},
year = {2004},
issn = {0022-1236},
doi = {https://doi.org/10.1016/j.jfa.2004.01.002},
url = {https://www.sciencedirect.com/science/article/pii/S0022123604000497},
author = {M Gubinelli}}

@book{book5,
place = {Newburyport},
title = {Optimal control: Linear quadratic methods},
publisher = {Dover Publications},
author = {Anderson, Brian D. O. and Moore, John B.},
year = {2014}}

@article{as,
title = {On some detection and estimation problems in heavy-tailed noise},
journal = {Signal Processing},
volume = {82},
number = {12},
pages = {1829-1846},
year = {2002},
issn = {0165-1684},
doi = {https://doi.org/10.1016/S0165-1684(02)00314-6},
url = {https://www.sciencedirect.com/science/article/pii/S0165168402003146},
author = {Ananthram Swami and Brian M Sadler}}

@inproceedings{jh2,
title = {Pink Noise {LQR}: How does Colored Noise affect the Optimal Policy in {RL}?},
author = {Jakob Hollenstein and Marko Zaric and Samuele Tosatto and Justus Piater},
booktitle = {ICML 2024 Workshop: Foundations of Reinforcement Learning and Control -- Connections and Perspectives},
year = {2024},
url = {https://openreview.net/forum?id=xRRWrxAKFL}}

@Inbook{tl2,
author = {Levajkovi{\'{c}}, Tijana and Mena, Hermann and Tuffaha, Amjad},
editor = {Oberguggenberger, Michael and Toft, Joachim and Vindas, Jasson and Wahlberg, Patrik},
title = {The Stochastic LQR Optimal Control with Fractional Brownian Motion},
booktitle = {Generalized Functions and Fourier Analysis: Dedicated to Stevan Pilipovi{\'{c}} on the Occasion of his 65th Birthday},
year = {2017},
publisher = {Springer International Publishing},
address = {Cham},
pages = {115--151},
isbn = {978-3-319-51911-1},
doi = {10.1007/978-3-319-51911-1_8},
url = {https://doi.org/10.1007/978-3-319-51911-1\_8}}

@article{jd,
title = {Stochastic control with rough paths},
volume = {75},
DOI = {10.1007/s00245-016-9333-9},
number = {2},
journal = {Applied Mathematics \& Optimization},
author = {Diehl, Joscha and Friz, Peter K. and Gassiat, Paul},
year = {2016},
month = {Feb},
pages = {285–315}}

@article{ted2,
author = {Duncan, Tyrone E. and Pasik-Duncan, Bozenna},
title = {Linear-Quadratic Fractional Gaussian Control},
journal = {SIAM Journal on Control and Optimization},
volume = {51},
number = {6},
pages = {4504-4519},
year = {2013},
doi = {10.1137/120877283},
URL = {https://doi.org/10.1137/120877283},
eprint = {https://doi.org/10.1137/120877283}}

@article{jml,
title = {Optimal control and filtering for systems with colored measurement noises},
DOI = {10.2514/6.1989-3538},
journal = {Guidance, Navigation and Control Conference},
author = {LIN, JIUM-MING},
year = {1989},
month = {Aug}}

@INPROCEEDINGS{fas,
author = {Shirazi, Farzad A. and Mohammadi, Masoud and Rahimi-Kian, Ashkan},
booktitle = {2007 International Conference on Mechatronics and Automation}, 
title = {Optimal Vibration Control of a Plate in Presence of Colored Noise Disturbances}, 
year = {2007},
pages = {2628-2633},
doi = {10.1109/ICMA.2007.4303971}}

@INPROCEEDINGS{ted3,
author = {Duncan, Tyrone E. and Pasik-Duncan, Bozenna},
booktitle = {2024 10th International Conference on Control, Decision and Information Technologies (CoDIT)}, 
title = {Prediction and Related Topics for a Scalar Linear Stochastic Equation with a Rosenblatt Process Noise}, 
year = {2024},
pages = {1475-1478},
doi = {10.1109/CoDIT62066.2024.10708566}}

@article{bet,
title = {Stochastic LQR optimal control with white and colored noise: Dynamic Programming Technique},
volume = {20},
DOI = {10.24275/rmiq/sim2353},
number = {2},
journal = {Revista Mexicana de Ingeniería Química},
author = {Escobedo-Trujillo, B. and Garrido-Meléndez, J.},
year = {2021},
month = {May},
pages = {1113–1129}}

@article{ted4,
title = {Advances in noise modeling for stochastic systems in Optimal Control},
volume = {1},
DOI = {10.1007/s44007-022-00025-y},
number = {3},
journal = {La Matematica},
author = {Duncan, Tyrone E. and Pasik-Duncan, Bozenna},
year = {2022},
month = {Apr},
pages = {685–695}}

@article{yh,
author = {Hu, Yaozhong and Oksendal, Bernt},
title = {Partial Information Linear Quadratic Control for Jump Diffusions},
journal = {SIAM Journal on Control and Optimization},
volume = {47},
number = {4},
pages = {1744-1761},
year = {2008},
doi = {10.1137/060667566},
URL = {https://doi.org/10.1137/060667566}}

@article{yh2,
author = {Hu, Yaozhong and Zhou, Xun Yu},
title = {Stochastic Control for Linear Systems Driven by Fractional Noises},
journal = {SIAM Journal on Control and Optimization},
volume = {43},
number = {6},
pages = {2245-2277},
year = {2005},
doi = {10.1137/S0363012903426045},
URL = {https://doi.org/10.1137/S0363012903426045},
eprint = {https://doi.org/10.1137/S0363012903426045}}

@article{st,
author = {Tang, Shanjian and Li, Xunjing},
title = {Necessary Conditions for Optimal Control of Stochastic Systems with Random Jumps},
journal = {SIAM Journal on Control and Optimization},
volume = {32},
number = {5},
pages = {1447-1475},
year = {1994},
doi = {10.1137/S0363012992233858},
URL = {https://doi.org/10.1137/S0363012992233858},
eprint = {https://doi.org/10.1137/S0363012992233858}}

@article{ct,
author = {Constantin Tudor and Maria Tudor},
title = {Approximation of Multiple Stratonovich Fractional Integrals},
journal = {Stochastic Analysis and Applications},
volume = {25},
number = {4},
pages = {781--799},
year = {2007},
publisher = {Taylor \& Francis},
doi = {10.1080/07362990701419979},
URL = {https://doi.org/10.1080/07362990701419979},
eprint = {https://doi.org/10.1080/07362990701419979}}

@article{mlk,
author = {Kleptsyna, M. L. and Breton, Alain Le and Viot, M.},
title = {About the linear-quadratic regulator problem under a fractional brownian perturbation},
journal = {ESAIM: Probability and Statistics},
pages = {161--170},
publisher = {EDP-Sciences},
volume = {7},
year = {2003},
doi = {10.1051/ps:2003007},
mrnumber = {1956077},
zbl = {1030.93059},
language = {en},
url = {https://www.numdam.org/articles/10.1051/ps:2003007}}

@article{mlk2,
author = {Kleptsyna, Marina L. and Breton, Alain Le and Viot, Michel},
title = {On the infinite time horizon linear-quadratic regulator problem under a fractional brownian perturbation},
journal = {ESAIM: Probability and Statistics},
pages = {185--205},
publisher = {EDP-Sciences},
volume = {9},
year = {2005},
doi = {10.1051/ps:2005008},
mrnumber = {2148966},
zbl = {1136.93463},
language = {en},
url = {https://www.numdam.org/articles/10.1051/ps:2005008}}

@article{continuity,
author = {D. Crisan and J. Diehl and P. K. Friz and H. Oberhauser},
title = {Robust filtering: Correlated noise and multidimensional observation},
volume = {23},
journal = {The Annals of Applied Probability},
number = {5},
publisher = {Institute of Mathematical Statistics},
pages = {2139 -- 2160},
year = {2013},
doi = {10.1214/12-AAP896},
URL = {https://doi.org/10.1214/12-AAP896}}

@article{turbulence1,
title = {The local structure of turbulence in incompressible viscous fluid for very large Reynolds numbers},
volume = {434},
DOI = {10.1098/rspa.1991.0075},
number = {1890},
journal = {Proceedings of the Royal Society of London. Series A: Mathematical and Physical Sciences},
author = {Kolmogorov, Andrei Nikolaevich},
year = {1991},
month = {Jul},
pages = {9–13}}

@article{turbulence2,
title = {Dissipation of energy in the locally isotropic turbulence},
volume = {434},
DOI = {10.1098/rspa.1991.0076},
number = {1890},
journal = {Proceedings of the Royal Society of London. Series A: Mathematical and Physical Sciences},
author = {Kolmogorov, Andrei},
year = {1991},
month = {Jul},
pages = {15–17}}

@article{rainfall,
author = {H. E. Hurst},
title = {Long-Term Storage Capacity of Reservoirs},
journal = {Transactions of the American Society of Civil Engineers},
volume = {116},
number = {1},
pages = {770-799},
year = {1951},
doi = {10.1061/TACEAT.0006518},
URL = {https://ascelibrary.org/doi/abs/10.1061/TACEAT.0006518},
eprint = {https://ascelibrary.org/doi/pdf/10.1061/TACEAT.0006518}}

@article{financial,
ISSN = {00219398, 15375374},
URL = {http://www.jstor.org/stable/2350970},
author = {Benoit Mandelbrot},
journal = {The Journal of Business},
number = {4},
pages = {394--419},
publisher = {University of Chicago Press},
title = {The Variation of Certain Speculative Prices},
urldate = {2025-08-30},
volume = {36},
year = {1963}}

@article{biology,
author = {T.T. Marquez-Lago  and A. Leier  and K. Burrage },
title = {Anomalous diffusion and multifractional Brownian motion: simulating molecular crowding and physical obstacles in systems biology},
journal = {IET Systems Biology},
volume = {6},
issue = {4},
pages = {134-142},
year = {2012},
doi = {10.1049/iet-syb.2011.0049},
URL = {https://digital-library.theiet.org/doi/abs/10.1049/iet-syb.2011.0049},
eprint = {https://digital-library.theiet.org/doi/pdf/10.1049/iet-syb.2011.0049}}

@article{quantum,
author = {Beenakker, Carlo and Schönenberger, Christian},
title = {Quantum Shot Noise},
journal = {Physics Today},
volume = {56},
number = {5},
pages = {37-42},
year = {2003},
month = {05},
issn = {0031-9228},
doi = {10.1063/1.1583532},
url = {https://doi.org/10.1063/1.1583532},
eprint = {https://pubs.aip.org/physicstoday/article-pdf/56/5/37/16723974/37\_1\_online.pdf}}

@article{lyons2,
title = {DIFFERENTIAL EQUATIONS DRIVEN BY ROUGH SIGNALS (I): AN EXTENSION OF AN INEQUALITY OF L. C. YOUNG},
author = {Terry Lyons},
journal = {Mathematical Research Letters},
year = {1994},
volume = {1},
pages = {451-464},
url = {https://api.semanticscholar.org/CorpusID:120979921}}

@article{clark,
title = {The design of robust approximations to the stochastic differential equations of nonlinear filtering},
DOI = {10.1007/978-94-011-7577-7_41},
journal = {Communication Systems and Random Process Theory},
author = {Clark, J. M.},
year = {1978},
pages = {721–734}}

@article{allan,
URL = {https://www.jstor.org/stable/26966008},
author = {Andrew L. Allan and Samuel N. Cohen},
journal = {The Annals of Applied Probability},
number = {5},
pages = {2274--2310},
publisher = {Institute of Mathematical Statistics},
title = {PATHWISE STOCHASTIC CONTROL WITH APPLICATIONS TO ROBUST FILTERING},
urldate = {2025-11-06},
volume = {30},
year = {2020}}

@book{operators,
place = {Cambridge},
edition = {2},
series = {Cambridge Studies in Advanced Mathematics},
title = {Lévy Processes and Stochastic Calculus},
publisher = {Cambridge University Press},
author = {Applebaum, David},
year = {2009},
collection = {Cambridge Studies in Advanced Mathematics}}

@article{bank2025stochastic,
  title={Stochastic control with signatures},
  author={Bank, Peter and Bayer, Christian and Hager, Paul P and Riedel, Sebastian and Nauen, Tobias},
  journal={SIAM Journal on Control and Optimization},
  volume={63},
  number={5},
  pages={3189--3218},
  year={2025},
  publisher={SIAM}
}

@article{bayer2023optimal,
  title={Optimal stopping with signatures},
  author={Bayer, Christian and Hager, Paul P and Riedel, Sebastian and Schoenmakers, John},
  journal={The Annals of Applied Probability},
  volume={33},
  number={1},
  pages={238--273},
  year={2023},
  publisher={Institute of Mathematical Statistics}
}

@article{lew2025rough,
  title={Rough Stochastic Pontryagin Maximum Principle and an Indirect Shooting Method},
  author={Lew, Thomas},
  journal={arXiv preprint arXiv:2502.06726},
  year={2025}
}

@article{ashkarian2026pontryagin,
  title={The Pontryagin maximum principle and $ Q $-functions in rough environments},
  author={Ashkarian, Estepan and Chakraborty, Prakash and Honnappa, Harsha and Tindel, Samy},
  journal={arXiv preprint arXiv:2601.05354},
  year={2026}
}

@article{chakraborty2024pathwise,
  title={Pathwise Relaxed Optimal Control of Rough Differential Equations},
  author={Chakraborty, Prakash and Honnappa, Harsha and Tindel, Samy},
  journal={arXiv preprint arXiv:2402.17900},
  year={2024}
}

\end{document}